\documentclass[11pt,a4paper]{article} 
\usepackage{xcolor}
\usepackage{lipsum}
\usepackage{setspace}
\usepackage{lineno}

\usepackage{geometry}
\usepackage{epsfig}
\usepackage{xcolor}
\usepackage[utf8x]{inputenc} 
\usepackage{hyperref}
\usepackage{amsmath,amssymb}
\def\Box{\rule{2mm}{2mm}}

\newtheorem{claim}{Claim}
\newtheorem{lemma}{Lemma}
\newtheorem{theorem}{Theorem}

\newcommand{\junk}[1]{}

\newenvironment{proof}{\noindent {\it Proof.}}{\Box \vskip \belowdisplayskip}

\begin{document}

\title{An Approximate Generalization of the Okamura-Seymour Theorem \footnote{A preliminary version of this paper appeared in the proceedings of the conference Foundations of Computer Science (FOCS) 2022 held in Denver, USA. }}
\author{Nikhil Kumar \\University of Waterloo}
\date{}
\maketitle

\begin{abstract} 
We consider the problem of multicommodity flows in planar graphs. Okamura and Seymour \cite{okamura1981multicommodity} showed that if all the demands are incident on one face, then the cut-condition is sufficient for routing demands. We consider the following generalization of this setting and prove an approximate max flow-min cut theorem: for every demand edge, there exists a face containing both its end points. We show that the cut-condition is sufficient for routing $\Omega(1)$-fraction of all the demands. To prove this, we give a $L_1$-embedding of the planar metric which approximately preserves distance between all pair of points on the same face.
\end{abstract}
\section{Introduction}
Given a graph $G$ with edge capacities and multiple source-sink pairs, each with an associated demand, the multicommodity flow problem is to route all demands simultaneously without violating edge capacities. The problem was first formulated in the context of VLSI routing in the 70's and since then it has seen a long and impressive line of work.

The demand graph $H$ is the graph obtained by including an edge $(s_{i},t_{i})$ for a demand with source-sink $s_{i},t_{i}$. A necessary condition for the flow to be routed is that the capacity of every cut exceeds the demand across the cut. This condition is called the $\textbf{cut-condition}$ and is known to be sufficient when $G$ is planar and all the source-sink pairs are on one face \cite{okamura1981multicommodity} or when $G+H$ is planar \cite{seymour1981odd}. However, one can construct small instances where the cut-condition is not sufficient for routing flow. When $G$ is series-parallel, if every cut has capacity at least twice the demand across it, then flow is routable \cite{chakrabarti2008embeddings,gupta2004cuts}. The flow-cut gap of a certain graph class is the smallest $\alpha$ such that flow is routable when capacity of every cut is at least $\alpha$ times the demand across it. Thus, for series-parallel graphs, the flow-cut gap is 2. For general graphs, the flow-cut gap is $\Theta(\log k)$ \cite{Linial1995}, where $k$ is the number of demand pairs.

The flow-cut gap for planar graphs ($G$ planar, $H$ arbitrary) is $\mathcal{O}(\sqrt{\log n})$ \cite{rao1999small} and is conjectured to be $\mathcal{O}(1)$ \cite{gupta2004cuts}. Chekuri et al.~\cite{chekuri2006embedding} showed a flow-cut gap of $2^{\mathcal{O}(k)}$ for $k$-outerplanar graphs. Lee et al.~\cite{krauthgamer2019flow} made progress towards this conjecture by showing an $\mathcal{O}(\log h)$ bound on the flow-cut gap, where $h$ is the number of faces on which source-sink vertices are incident. Filtser \cite{filtser2020face} further improved this bound by showing a flow-cut gap of $\mathcal{O}(\sqrt{\log h})$, when all the source-sink vertices are incident on $h$ faces. In this paper, we consider instances where the source and sink of each demand lie on the same face, but all the source-sink pairs don't necessarily lie on a single face, and show that the flow-cut gap of such instances is $\mathcal{O}(1)$. It is well known that the cut-condition is not sufficient for such instances (see Figure \ref{intro_cc_not_suff}).
\begin{figure}[ht] \label{intro_cc_not_suff}
\centering
\includegraphics[width=6 in]{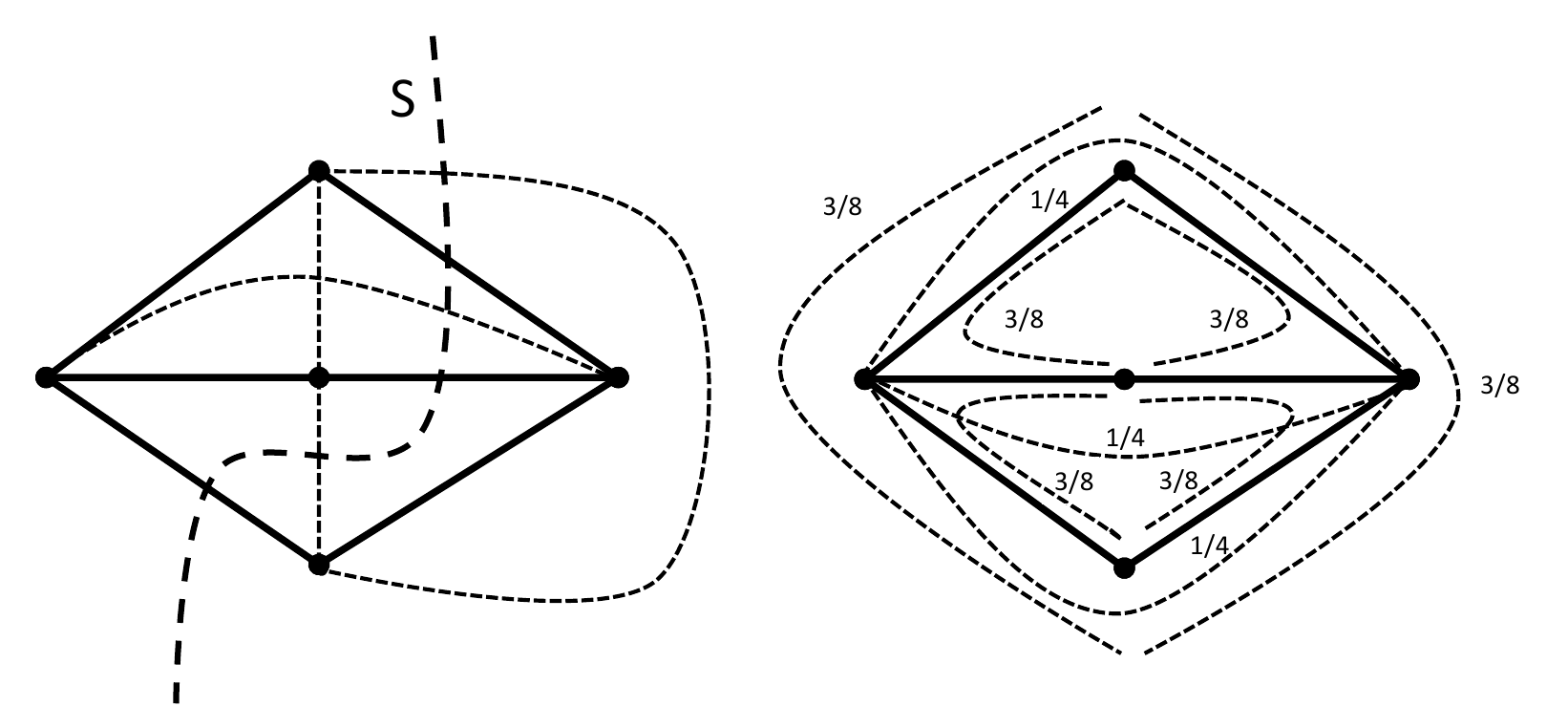}
\caption{This example first appeared in the work of Okamura and Seymour~\cite{okamura1981multicommodity}. All supply (solid) and demand (dashed) edges have value 1. $S$ (bold-dashed) is a cut. The total capacity of supply edges across $S$ is three while $S$ separates three units of demand; hence $S$ satisfies the cut-condition. One can check that no cut violates the cut-condition. Since the source-sink of every demand is distance two apart, a total capacity of $4 \cdot 2=8$ is required for a feasible routing, but only six are available. Hence, no feasible routing is possible. This also implies that no more than 3/4 of every demand can be routed simultaneously. The figure on the right shows a feasible routing of 3/4 of every demand, which implies a flow-cut gap of 4/3.}
\end{figure}

A common approach to establish bounds on the flow-cut gap is to bound the $L_{1}$ distortion incurred in embedding an arbitrary metric on the graph $G$ into a normed space. This, for instance, has been the method used to establish flow-cut gaps for general graphs \cite{Linial1995}, series-parallel graphs \cite{gupta2004cuts,chakrabarti2008embeddings} and planar graphs \cite{rao1999small}. We too build on this technique to prove our results (see Theorem \ref{Theorem : Main} and \ref{Theorem : Main_embedding}).

\section{Definitions and Preliminaries}
Let $G=(V,E)$ be a simple graph with edge capacities $c:E \rightarrow R_{\geq 0}$. We call this the supply graph. Let $H=(V,D)$ be a simple graph with demands on edges $d:D \rightarrow R_{\geq 0}$. We call this the demand graph. The objective of the $\bf{multicommodity}$ $\bf{flow}$ problem is to find paths between the end points of demand edges in the supply graph such that the following hold: for every demand edge $e \in D$, $d(e)$ paths are picked in the supply graph and every supply edge $e\in E$ is present in at most $c(e)$ paths.

We say that an instance is integrally feasible if paths satisfying the above two conditions can be found. If there exists an assignment of non-negative real numbers to paths such that total flow for every demand edge $e \in D$ is $d_{e}$ and total value of paths using a supply edge $e \in E$ is at most $c_{e}$, then we say that the instance is (fractionally) feasible. For $S \subseteq V$, a cut $(S,V \setminus S)$ is a bi-partition of the vertex set. We refer to the cut $(S,V \setminus S)$ by $S$ for convenience. The total capacity of supply edges of $G$ going across the cut $S$ is denoted by $\delta_{G}(S)$ (i.e.~the sum of capacities of supply edges having exactly one end-point in $S$). Similarly, $\delta_{H}(S)$ denotes the sum of demand values of demand edges going from $S$ to $V \setminus S$. One necessary condition for routing the flow (fractionally or integrally) is as follows: $\delta_{G}(S) \geq \delta_{H}(S)$ for every $S \subseteq V$. In other words, across every cut, total supply should be at least the total demand. This condition is also known as the \textbf{cut-condition}. In general, the cut-condition is not sufficient for a feasible routing. We can ask for the following relaxation: given an instance for which the cut-condition is satisfied, what is the maximum value of $\alpha$, such that $\alpha$ fraction of every demand can be routed? The number $\alpha^{-1}$ is known as the $\textbf{flow-cut gap}$ of the instance. There is an equivalent definition of the flow-cut gap: given an instance $(G,H)$ satisfying the cut condition, the smallest number $k$ such that $(kG,H)$ is feasible, where $kG$ denotes the graph with every edge capacity multiplied by $k$ (see Figure \ref{intro_cc_not_suff} for an illustration). We call $k$ the \textbf{congestion} of the multicomodity-flow instance $(G,H)$. The following classic result identifies a setting where the cut-condition is also sufficient for routing demands in planar graphs. We will be invoking these to prove our result.
\begin{theorem} [Okamura -Seymour~\cite{okamura1981multicommodity}]\label{OS}
If $G$ is a planar graph, all the edges of $H$ are restricted to a face and $G+H$ is Eulerian, then the cut-condition is necessary and sufficient for integral routing of all the demands.
\end{theorem}
\begin{theorem} [Seymour~\cite{seymour1981odd}] \label{seymour}
If $G+H$ is planar and Eulerian, then the cut-condition is necessary and sufficient for integral routing of all the demands.
\end{theorem}
In the above theorem statements, $G+H$ is Eulerian if the total value of supply and demand edges incident on any vertex is an even integer. Observe that if the Eulerian condition is not satisfied, then we obtain a half-integral flow (i.e.~all flow paths have a flow that is an integer multiple of 1/2), that can be computed in polynomial time. In this paper, we will not be concerned with (half) integral flows and focus on proving flow-cut gap for instances that generalize the setting of the above (see Theorem \ref{Theorem : Main}).

From now on, we assume a fixed planar embedding of $G$. Without loss of generality, one can assume that $G$ is 2-vertex connected. If there is a cut-vertex $v$ and $ab$ is a demand separated by removal of $v$, then replacing $ab$ by $av,vb$ maintains the cut-condition. By doing this for every cut-vertex and demand separated by them, we get separate smaller instances for each 2-vertex connected component. Hence, every vertex is a part of a cycle corresponding to some face. By our assumption, for every demand edge there exists a face such that both its end points lie on that face. Hence, we can associate every demand with a face. We abuse notation and use $f$ to also denote the edges and vertices associated with the cycle of face $f$. Given a set $S$, we denote the subgraph induced by vertices in $S$ as $G[S]$. We call a subset $A \subseteq V$ central if both $G[A]$ and $G[V-A]$ are connected. A cut $(A, V \setminus A) $ is called central if $A$ is central. The following is well-known.
\begin{lemma}[\cite{schrijver2003combinatorial}]\label{central}
$(G,H)$ satisfies the cut-condition if and only if $\delta_G(A) \geq \delta_H(A)$ for all the central sets $ A \subseteq V$.
\end{lemma}
The set of all faces of $G$ will be denoted by $F$. The $\bf{dual}$ of a planar graph, denoted by $G^{D}=(V^{D},E^{D})$, is defined as follows: $V^{D}= F$ and if $f_{i},f_{j} \in F$ share an edge in $G$, then $(f_{i},f_{j}) \in E^{D} $. It is a well known fact that edges of a central cut in $G$ correspond to a  simple circuit in $G^{D}$ and vice versa. Given a graph $G=(V,E)$ with edge-lengths $l:E \rightarrow \mathbb{R}_{\geq 0}$, we use $d_{G}(u,v)$ (if $l$ is clear from the context) and $l(u,v)$ (if $G$ is clear from the context) to denote the shortest path distance between $u$ and $v$ in $G$ w.r.t edge-length $l$. We now describe the connection between flow-cut gap and embedding vertices into normed spaces.
\subsection{Embedding Metrics into $L_1$}
Given an edge weighted graph $G=(V,E)$ with edge-lengths $l:E \rightarrow \mathbb{R}_{\geq 0}$, associated shortest path metric $d_{G}$, a graph $H=(V,D)$ and an embedding $z : V \rightarrow \mathbb{R}^{m}$, the contraction and expansion of $z$ are the smallest $\alpha$, $\beta$ respectively, such that $||z(u)-z(v)||_{1} \geq  d_{G}(u,v) / \alpha$ for all $(u,v) \in D$ and $ ||z(u)-z(v)||_{1} \leq \beta \cdot d_{G}(u,v)$ for all $(u,v) \in E$. The distortion of the embedding, $\texttt{dist}(G,H,z)$, is $\alpha \cdot \beta$. Given $G,H$, we are generally interested in finding an embedding with low distortion. If $H$ is a clique, we refer to $\texttt{dist}(G,H,z)$ simply as $\texttt{dist}(G,z)$ and call it the distortion of $G$ with respect to $z$. Linial et al.~\cite{Linial1995} built on the result of Bourgain \cite{bourgain1985lipschitz}, and gave a polynomial time algorithm that embeds any graph on $n$ vertices into $L_1$ with distortion $\mathcal{O}(\log n)$. Furthermore, this result is asymptotically the best possible, as there exist instances for which any embedding into $L_1$ has distortion $\Omega(\log n)$. There is a rich literature on finding low distortion embeddings for special graph classes. It is well known that a tree can be embedded into $L_1$ with distortion 1, outerplanar graphs with distortion 1 \cite{okamura1981multicommodity}, $k$-outerplanar graphs with distortion $2^{\mathcal{O}(k)}$ \cite{chekuri2006embedding}, series-parallel graphs with distortion 2 \cite{chakrabarti2008embeddings}. Rao \cite{rao1999small} showed that any planar metric can be embedded into $L_1$ with distortion $\mathcal{O}(\sqrt{\log n})$ and there has essentially been no improvement upon this in the last two decades. It is conjectured that any planar graph can be embedded into $L_1$ with distortion $\mathcal{O}(1)$ \cite{gupta2004cuts} (also known as the \texttt{GNRS-Conjecture}). In this paper, we make progress towards this conjecture. Given a drawing of a planar graph in the plane, let $H$ be the set of all pairs of vertices $(u,v)$ such that $u$ and $v$ lie on the same face. We show the existence of a polynomial time computable $z:V \rightarrow \mathbb{R}^{t}$ such that $\texttt{dist}(G,H,z)=\mathcal{O}(1)$. In this paper, we will work exclusively with the 1-norm, so we drop the subscript and denote $||z(u)-z(v)||_1$ simply as $||z(u)-z(v)||$. We say that an embedding is non-expansive if $\beta=1$. The dimension of an embedding is not of primary concern to us in this work, therefore we will use $z:V \rightarrow L_1$ to denote the fact that there is a mapping of elements of $V$ into $\mathbb{R}^t$ for some natural number $t$. By scaling, we may convert any $L_1$ embedding into a non-expansive one. Given two embeddings $z : V \rightarrow \mathbb{R}^{m}$ and $w: V \rightarrow \mathbb{R}^{l}$ and non-negative real numbers $\alpha, \beta$, we denote by $\alpha \cdot z + \beta \cdot w $ the embedding obtained by concatenating $z$ and $w$ after scaling each of their coordinates by $\alpha$ and $\beta$ respectively. More formally, if $z(u)=(z_1,z_2,\ldots, z_m)$ and $w(u)=(w_1,w_2,\ldots,w_l)$ for $ u \in V$, then $(\alpha \cdot z + \beta \cdot w) (u)=(\alpha \cdot z_1,\alpha \cdot z_2,\ldots, \alpha \cdot z_m, \beta \cdot w_1,\beta \cdot w_2,\ldots,\beta \cdot w_l)$. 
\subsection{Flow-Cut Gap and Embedding into $L_1$}
Flow-cut gaps and embedding metrics into $L_1$ are intimately related. This connection was first observed by Linial et al.~\cite{Linial1995}, who used it to prove flow-cut gap results. We describe this connection formally now. Let $G=(V,E), H=(V,D)$ be fixed graphs and $c:E \rightarrow \mathbb{R}_{\geq 0}$, $d:D \rightarrow \mathbb{R}_{\geq 0}$ and $l:E \rightarrow \mathbb{R}_{\geq 0}$ denote capacity, demand and lengths of respective edge sets. Let $\mathcal{I}$ be the set of all multicommodity flow instances $G=(V,E,c),H=(V,D,d)$ for which the cut-condition is satisfied. Note that we construct $\mathcal{I}$ by keeping $G,H$ fixed and considering all possible values of $c$ and $d$ for which the cut-condition is satisfied. Let $\texttt{cong}(G,H)$ denote the maximum congestion required for routing any multicommodity flow instance in $\mathcal{I}$. Let $G_l$ denote the graph $G$ with edge-length $l$. Let $\texttt{dist}(G,H)$ be the minimum number such that for every $l$, there exists a $f$ such that $\texttt{dist}(G_l,H,f) \leq \texttt{dist}(G,H)$. The $\textbf{congestion-distortion theorem}$ states that $\texttt{cong}(G,H)=\texttt{dist}(G,H)$. See Section 3 of \cite{chakrabarti2012cut} for a simple proof of this fact using linear programming duality. This connection has been exploited extensively to prove flow-cut gap results for general graphs \cite{Linial1995}, series-parallel graphs \cite{chakrabarti2008embeddings, gupta2004cuts} and planar graphs \cite{rao1999small}. All these results proceed by showing the existence of a low distortion embedding of the corresponding metric into $L_1$. Using the congestion-distortion theorem, we now restate the theorem of Okamura and Seymour~\cite{okamura1981multicommodity} and Seymour~\cite{seymour1981odd} in terms of metric embedding.

\begin{theorem} [Okamura-Seymour\cite{okamura1981multicommodity}]\label{OS_embedding}
Let $G=(V,E)$ be a planar graph with edge-lengths $l:E \rightarrow \mathbb{R}_{\geq 0}$ and $t \in F$ be one of its faces. Then there exists an embedding $z: V \rightarrow L_1$ such that $||z(u)-z(v)||=d_{G}(u,v)$ for all $u,v \in t$ and $||z(u)-z(v)|| \leq d_{G}(u,v)$ for all $(u,v) \in E$.
\end{theorem}

\begin{theorem} [Seymour\cite{seymour1981odd}] \label{seymour_embedding}
Let $G=(V,E)$ be a planar graph with edge-lengths $l:E \rightarrow \mathbb{R}_{\geq 0}$ and $H=(V,T)$ be a demand graph such that $G+H$ is planar. Then there exists an embedding $z: V \rightarrow L_1$ such that $||z(u)-z(v)||=d_{G}(u,v)$ for all $(u,v) \in T$ and $||z(u)-z(v)|| \leq d_{G}(u,v)$ for all $(u,v) \in E$.  
\end{theorem}

\subsection{Cut Metrics and $L_1$ Embedding} \label{Cut-Metric-Embedding Equivalence}

Suppose we have a set of cuts with non-negative weights $\mathcal{C}=\{(C_1,w_1),\ldots,(C_k,w_k)\}$. Define $C_i(u,v)$ to be $w_i$ if exactly one of $u,v$ is contained in $C_i$ and 0 otherwise. Let $\mathcal{C}(u,v)=\sum_{i=1}^{k}C_i(u,v)$. It is easy to verify that $\mathcal{C}$ induces a metric on $V$. We refer to the metric induced by $\mathcal{C}$ as a \textbf{cut-metric}. One can construct $f:V \rightarrow \mathbb{R}^{k}$ such that $\forall u,v$ we have $||f(u)-f(v)||_1=\mathcal{C}(u,v)$ as follows: for any vertex $u$, define $f(u)=(u_1,u_2,\ldots,u_k)$ where $u_i=w_i$ if $u \in C_i$, 0 otherwise. In fact, converse of the above is also true: given any embedding of vertices into $L_1$, there exists a set of weighted cuts $\mathcal{C}$ such that the distance metric induced by $\mathcal{C}$ is equal to the distance metric induced by the $L_{1}$ embedding (see Lemma 15.2 of \cite{williamson2011design} for a proof). Hence, to show a low distortion $L_1$ embedding of a metric, it is equivalent to show a collection of cuts which preserve distances with low distortion. Using the aforementioned equivalence of the cut-metric and $L_1$ embedding, we use them interchangeably from now on. Given a non-negative real number $\alpha$ and a collection of weighted cuts $\mathcal{C}$, $\alpha \cdot \mathcal{C}$ denotes the the same collection of cuts with the weight of all cuts scaled by a multiplicative factor of $\alpha$.

\section{Our Contribution}
We generalize the result of Okamura and Seymour \cite{okamura1981multicommodity} and prove the following approximate max flow-min cut theorem:
\begin{theorem} \label{Theorem : Main}
Let $G$ be an edge-capacitated planar graph and $H$ be a set of demand edges such that for each $(u,v) \in H$, there exists a face $f$ containing both $u$ and $v$. If the cut-condition is satisfied, then there exists a feasible routing of $\Omega(1)$-fraction of all the demands.
\end{theorem}
Using the congestion-distortion theorem, Theorem \ref{Theorem : Main} can be restated in terms of metric embedding as follows:
\begin{theorem} \label{Theorem : Main_embedding}
Let $G=(V,E)$ be a planar graph with edge-lengths $l \rightarrow \mathbb{R}_{\geq 0}$ and $T$ be pairs of vertices $(u,v)$ such that both $u$ and $v$ lie on the same face. Then there exists an absolute constant $c>1$ and an embedding $z: V \rightarrow L_1$ such that $||z(u)-z(v)|| \geq d_{G}(u,v)/c$ for $(u,v) \in T$ and $||z(u)-z(v)|| \leq d_G(u,v)$ for $(u,v) \in E$.
\end{theorem}
We now give a brief overview of our approach. As mentioned before, we work with a fixed embedding of the given planar graph in the plane. We call a face $f \in F$ \textbf{geodesic} if for all $u,v \in f$, $d_{G}(u,v)$ is equal to the shortest path distance between $u,v$ using only the edges of the cycle corresponding to $f$. A face which is not geodesic is called \textbf{non-geodesic}. Given a cycle $S$, let $R(S)$ (resp.~$R^{\circ}(S)$) be the closed (resp.~open) region bounded by $S$ in the plane. Given a  face $f$, let $S^{\infty}_{f}$ be a minimal cycle such that for any $u,v \in R(S^{\infty}_{f})$, there exists a shortest $u,v$ path contained completely inside $R(S^{\infty}_{f})$. If $f$ is a geodesic face, then $S_f^{\infty}$ is exactly the cycle bounding the face, i.e.~$S_f^{\infty}=f$ (see Section \ref{Section:laminar_struc_face_support} for precise definitions).

In Section \ref{Section:laminar_struc_face_support}, we show that the set $\{R(S^{\infty}_f)| f \in F\}$ forms a laminar (or non-crossing) set system. This implies that there is a non-geodesic face $f$ with minimal $R(S^{\infty}_f)$, i.e.~for any non-geodesic face $f'\neq f$ either $R(S^{\infty}_f) \subseteq R(S^{\infty}_{f'})$ or $R^{\circ}(S^{\infty}_f) \cap R^{\circ}(S^{\infty}_{f'})= \emptyset$. We inductively find a suitable embedding of the graph obtained by removing the vertices in the interior of $R(S^{\infty}_f)$ and show how to extend the embedding to include all the vertices contained in $R(S^{\infty}_f)$. This extension argument turns out to be non-trivial and forms the core of our proof. To do this extension, we need to develop several new tools. In Section \ref{Section:laminar_struc_face_support}, we give an algorithm to modify the original length function so as to allow a nice inductive decomposition. We call such length functions $\alpha$-good and believe that this could be a useful tool in proving flow-cut gaps for other planar instances as well. In Section \ref{Section: embed_geodesic_pairs}, we come up with an embedding for all geodesic pairs of vertices, i.e.~pairs of vertices for which there exists a shortest path using only the edges on the corresponding face. In Section \ref{Section:single_source}, we develop a low distortion embedding for all pairs of vertices whose shortest path uses a fixed vertex of the graph. We believe that this problem is interesting in its own right and could prove to be useful in other settings. In Sections \ref{Section:constrained_embedding} and \ref{Section:final}, we combine the tools developed in previous sections to complete the proof of main theorem. 

\section{Laminar Structure of Face Supports}  \label{Section:laminar_struc_face_support}
We partition the of faces of planar graph $G$ into two sets: $\bf{geodesic}$ and $\textbf{non-geodesic}$. A face $f$ is called geodesic if for all $u,v \in f$, there exists a shortest path between $u$ and $v$ using only the vertices of the cycle associated with $f$. A face that is not geodesic is called a non-geodesic face. Let $f_I$ be the infinite face in the planar embedding of $G$. Let $F_N$ denote the set of all the non-geodesic faces except $f_I$. Let $F_G$ be the set of all the geodesic faces.

A set of edges $E_S \subseteq E$ is called a \textbf{support} of $S \subseteq V$ if for all $u,v \in S$, $d_{G_S}(u,v)=d_{G}(u,v)$, where $G_S=(V,E_S)$. In other words, restricting to $E_S$ does not change the induced shortest path metric on $S$. If $d_{G}(u,v) < l(u,v)$, then we may remove the edge $(u,v)$ without changing the shortest path metric on $G$. Hence we may assume that $d_{G}(u,v)=l(u,v)$ for all $(u,v) \in E$.  Given a cycle $C$, let $R^{\circ}(C)$ (resp.~$R(C)$) denote the open (resp.~closed) region contained inside the cycle $C$ in the planar embedding of $G$. We choose $R^{\circ}(C)$ and $R(C)$ such that it does not contain the infinite face $f_I$. For a cycle $C$, we overload notation and use $R(C)$ (resp.~$R^{\circ}(C)$) to denote the set of vertices contained inside $R(C)$ (resp.~$R^{\circ}(C)$). Note that the vertices of $C$ are a part of $R(C)$ but not $R^{\circ}(C)$. For $S \subseteq V$, let $G[S]$ denote the induced graph on $S$ and $G \setminus S$ denote the induced graph on $V\setminus S$.



Let $S^{1}_f$ be a cycle such that $R(S^{1}_f)$ is the inclusion wise minimal (closed) region containing a support of $f$, i.e.~there exists a support $E_f$ of $f$ such that $u,v \in R(S^{1}_f)$ for all $(u,v) \in E_f$. There could be multiple candidates for $S^{1}_f$, we fix one of them arbitrarily. Note that $R(f) \subseteq R(S_f^1)$ and if $f \in F_G$, then $S^{1}_f=f$. For $i >1$ let $S^{i+1}_f$ be a cycle such that $R(S^{i+1}_f)$ is the inclusion wise minimal region containing a support of the face (corresponding to the cycle) $S^{i}_f$ in the graph $G \setminus R^{\circ}(S^{i}_f)$. Let $S_f^{\infty}= \displaystyle\lim_{i\to\infty} S^{i}_f $.

Note that $S^i_f$ is a face in the graph $G \setminus R^{\circ}(S^{i}_f)$. If the face corresponding to $S^{i}_f$ in the graph $G \setminus R^{\circ}(S^{i}_f)$ is geodesic, then $S_f^\infty=S^{i}_f$. If the face corresponding to $S^{i}_f$ in the graph $G \setminus R^{\circ}(S^{i}_f)$ is non-geodesic, then $R(S^{i}_f) \subset R(S^{i+1}_f)$. Hence, $S_f^\infty =S^{n}_f$ and we can compute $S_f^\infty$ in polynomial time. We will show that $\{R(S^{\infty}_f)\}_{f \in F}$ form a laminar family. We first note the following claim, which will be useful later.
\begin{claim} \label{claim:S_f_geodesic}
     The face corresponding to the cycle $S_f^\infty$ in the graph $G \setminus R^{\circ}(S_f^\infty)$ is geodesic. Furthermore, $d_{G[R(S_f^\infty)]}(u,v)=d_G(u,v)$ for all $u, v \in R(S_f^\infty)$.
\end{claim}
\begin{proof}
The first part of the claim follows immediately from the discussion above. Let $P[u,v]$ be a shortest $u,v$ path for some $u,v \in R(S_f^\infty)$. If all the vertices in $P[u,v]$ are contained inside $R(S_f^\infty)$, then $d_{G[R(S_f^\infty)]}(u,v)=d_G(u,v)$. If not, then $P[u,v]$ uses some vertices not in $R(S_f^\infty)$. We can write $P[u,v]$ as a union of three paths: $P[u,x],P[x,y],P[y,v]$, where $P[u,x]$ and $P[y,v]$ are contained completely inside $R(S_f^\infty)$. Note that $x,y \in S_f^\infty$. Since $S_f^\infty$ is geodesic in the graph $G \setminus R^{\circ}(S_f^\infty)$, there exists a shortest $x,y$ path $P'[x,y]$ along the cycle $S_f^\infty$. Hence the union of paths $P[u,x],P'[x,y],P[y,v]$ is a shortest $u,v$ path contained completely inside $R(S_f^\infty)$ and the claim follows. 
\end{proof} 
Let $f$ be a face of $G$ and $u,v \in f$. Let $S(f,u,v)$ be a cycle such that $R(S(f,u,v)) \subseteq R(S_1^f)$ is a minimal (closed) region containing the face $f$ and a shortest $u,v$ path. For convenience, we will refer to $R(S(f,u,v))$ (resp.~$R^{\circ}(S(f,u,v))$) simply as $R(f,u,v)$ (resp.~$R^{\circ}(f,u,v)$). We note that there could be multiple choices for $R(f,u,v)$, we fix one of them arbitrarily.
\begin{claim} \label{clm:minimality}
    For any $f \in F$ and $u,v \in f$, $S(f,u,v) \subseteq E(f) \cup P[u,v]$, where $E(f)$ is the set of edges of face $f$ and $P[u,v]$ is the unique shortest $u,v$ path contained inside $R(f,u,v)$. 
\end{claim}

\begin{proof}
Consider a shortest $u,v$ path $P[u,v]$ in $R(f,u,v)$. Then $R(f,u,v)$ must be contained in the closed region bounded by the face $f$ and $P[u,v]$ (by definition). This implies that $P[u,v]$ forms a part of the boundary of $R(f,u,v)$ and hence $P[u,v] \subseteq S(f,u,v) \subseteq E(f) \cup P[u,v]$. The existence of a shortest $u,v$ path other than $P[u,v]$ contained inside $R(f,u,v)$ would contradict the minimality of $R(f,u,v)$. Hence there is a unique $u,v$ shortest path inside $R(f,u,v)$.
\end{proof}

\begin{lemma} \label{lemma:laminar rfuv}
Let $f_1, f_2 \in F$ be distinct faces of $G$ such that $u_1,v_1 \in f_1$ and $u_2,v_2 \in f_2$. Then one of the following must hold: $R^{\circ}(f_1,u_1,v_1) \cap R^{\circ}(f_2,u_2,v_2)=\emptyset$ or $R(f_1,u_1,v_1) \subseteq R(f_2,u_2,v_2)  \setminus R^{\circ}(f_2)$ or $R(f_2,u_2,v_2) \subseteq R(f_1,u_1,v_1) \setminus R^{\circ}(f_1)$.
\end{lemma}
\begin{proof}
Let $P[u_1,v_1]$ be the unique $u_1,v_1$ shortest path contained inside $R(f_1,u_1,v_1)$. Suppose $R^{\circ}(f_1) \cap R^{\circ}(f_2,u_2,v_2)=\emptyset$. If none of the vertices of $P[u_1,v_1]$ are in $R^{\circ}(f_2,u_2,v_2)$, then $R^{\circ}(f_1,u_1,v_1) \cap R^{\circ}(f_2,u_2,v_2)=\emptyset$ and the condition of the lemma is satisfied. Suppose that this is not the case and the path $P[u_1,v_1]$ contains vertices of $R^{\circ}[f_2,u_2,v_2]$. Order the vertices on the path $P[u_1,v_1]$ from $u_1$ to $v_1$ and choose $x, y$ on it such that (i) $x,y \in S(f_2,u_2,v_2)$ (ii) no vertex of $S(f_2,u_2,v_2)$ appears between $x$ and $y$ in the ordering. Let $P_1[x,y] \subseteq S(f_2,u_2,v_2)$ be a $x,y$ path such that none of the vertices of $P_1[x,y]$ is on face $f_2$ (see Figure \ref{fig:lemma_2}). By Claim \ref{clm:minimality} such a path exists and $P_1[x,y]$ is a shortest $x,y$ path. We can replace $P[x,y]$ by $P_1[x,y]$ to obtain another shortest $u_1,v_1$ path contained inside $R(f_1,u_1,v_1)$. This is a contradiction as $P[u_1,v_1]$ is the unique $u_1,v_1$ path contained inside $R(f_1,u_1,v_1)$ (by Claim~\ref{clm:minimality}).

Now suppose that $R(f_1) \subseteq R(f_2,u_2,v_2)$. If all the vertices in $P[u_1,v_1]$ are contained inside the region $R(f_2,u_2,v_2)$, then $R(f_1,u_1,v_1) \subseteq R(f_2,u_2,v_2)  \setminus R^{\circ}(f_2)$ and the condition of the lemma is satisfied. Suppose that this is not the case and the path $P[u_1,v_1]$ has vertices that are not contained in $R[f_2,u_2,v_2]$. Order the vertices on the path $P[u_1,v_1]$ from $u_1$ to $v_1$ and choose $x, y$ on it such that (i) $x,y \in S(f_2,u_2,v_2)$ (ii) no vertex of $S(f_2,u_2,v_2)$ appears between $x$ and $y$ in the ordering (see Figure~\ref{fig:lemma_2}). Let $P_1[x,y] \subseteq S(f_2,u_2,v_2)$ be a $x,y$ path such that none of the vertices of $P_1[x,y]$ is on face $f_2$. By Claim \ref{clm:minimality} such a path exists and $P_1[x,y]$ is a shortest $x,y$ path. We can replace $P[x,y]$ by $P_1[x,y]$ to obtain another shortest $u_1,v_1$ path contained inside $R(f_1,u_1,v_1)$. This is a contradiction as $P[u_1,v_1]$ is the unique $u_1,v_1$ path contained inside $R(f_1,u_1,v_1)$ (by Claim~\ref{clm:minimality}).
\end{proof}

\begin{figure}[ht!]  
\centering
\includegraphics[width=6.5 in]{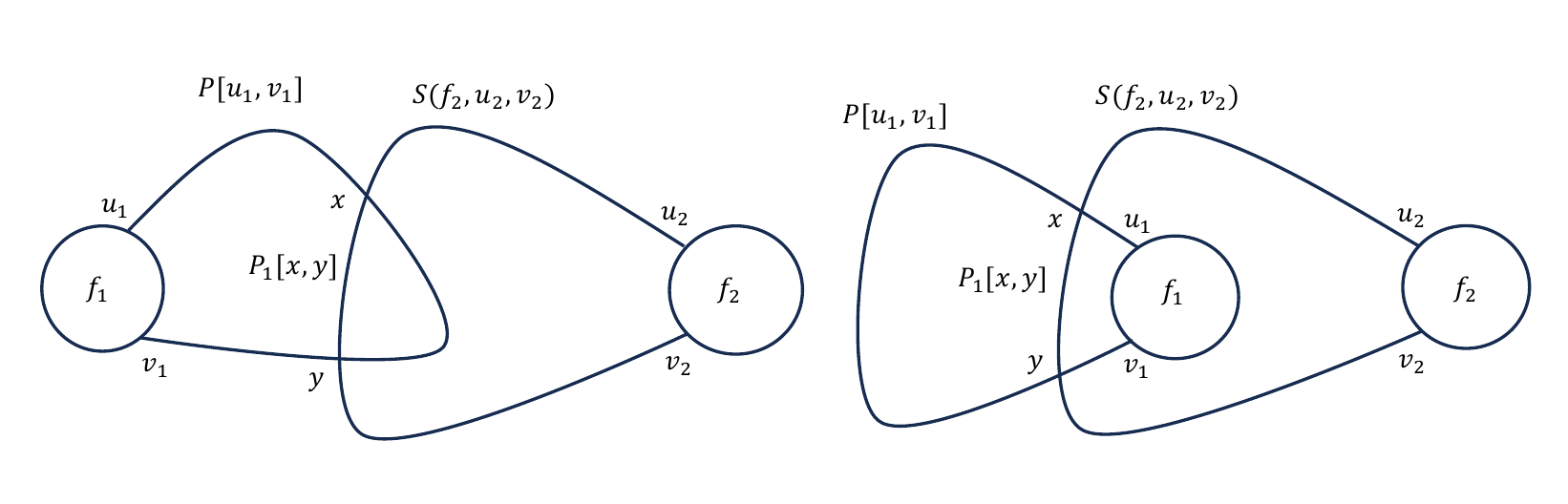}
\caption{Illustrations for the proof of Lemma \ref{lemma:laminar rfuv}.}
\label{fig:lemma_2}
\end{figure}

\begin{claim} \label{clm:sf_rfuv}
     $R(S_f^1)=\bigcup_{u,v \in f} R(f,u,v)$ for each $f \in F$.
\end{claim}         

\begin{proof}
    By the definition of $R(f,u,v)$, it follows that $\bigcup_{u,v \in f} R(f,u,v) \subseteq R(S_{f}^{1})$. Since $R(S_f^1)$ is a minimal region containing at least one shortest path between every pair of vertices in $f$, we have $R(S_f^1) \subseteq \bigcup_{u,v \in f} R(f,u,v)$, and the claim follows. 
\end{proof}

\begin{lemma} \label{Lemma: face_support_laminar}
The following are true for any distinct $f_1,f_2 \in F$:

\begin{enumerate}
    \item Either $R^\circ(S^1_{f_1}) \cap R^\circ(S^1_{f_2})= \emptyset$ or $R(S^1_{f_1}) \subseteq R(S^1_{f_2})$ or $R(S^1_{f_2}) \subseteq R(S^1_{f_1})$. 

    \item If $R(f _1) \subseteq R(S^1_{f_2})$, then $R(S^1_{f_1}) \subseteq R(S^1_{f_2}) \setminus R^{\circ}(f_2)$. Furthermore, $R(S^1_{f_1}) \neq R(S^1_{f_2})$.
\end{enumerate}
\end{lemma}
\begin{proof}
Suppose there exist $u_1,v_1 \in f_1$ such that $R(f_2) \subseteq R(f_1,u_1,v_1)$. By Lemma~\ref{lemma:laminar rfuv}, for all $u_2, v_2 \in f_2$, $R(f_2,u_2,v_2) \subseteq R(f_1,u_1,v_1) \setminus R^{\circ}(f_1)$. By Claim~\ref{clm:sf_rfuv}, this implies that $R(S^1_{f_2}) \subseteq R(S^1_{f_1})\setminus R^{\circ}(f_1)$. The case when there exist $u_2,v_2 \in f_2$ such that $R(f_1) \subseteq R(f_2,u_2,v_2)$ is symmetric and can be handled by the same argument as above. Suppose that for all $u_1,v_1 \in f_1$ and $u_2,v_2 \in f_2$, we have that $R^{\circ}(f_1,u_1,v_1) \cap R^{\circ}(f_2,u_2,v_2)=\emptyset$. By using Claim~\ref{clm:sf_rfuv}, this implies that $R^\circ(S^1_{f_1}) \cap R^\circ(S^1_{f_2})=\emptyset$. Hence one of the conditions stated in the first part of the lemma must always hold. 

The proof of part two of the lemma is similar to part one. If $R(f _1) \subseteq R(S^1_{f_2})$, then by Claim~\ref{clm:sf_rfuv} there must exist $u_2,v_2 \in f_2$ such that $R(f_1) \subseteq R(f_2,u_2,v_2)$. By Lemma~\ref{lemma:laminar rfuv}, for all $u_1, v_1 \in f_1$, $R(f_1,u_1,v_1) \subseteq R(f_2,u_2,v_2) \setminus R^{\circ}(f_2)$. By Claim~\ref{clm:sf_rfuv}, this implies that $R(S^1_{f_1}) \subseteq R(S^1_{f_2})\setminus R^{\circ}(f_2)$. Observe that if $R(S^1_{f_1})=R(S^1_{f_2})$, then $R(f_1) \subseteq R(S^1_{f_2})$ and hence $R(S^1_{f_1}) \subseteq R(S^1_{f_2}) \setminus R^{\circ}(f_2) \subset R(S^1_{f_2})$, which is a contradiction.
\end{proof}
\begin{lemma} \label{Lemma:no_internal_path}
Let $u,v \in S^1_f$ for some $f \in F$ and $P=\{u,u_1,u_2,\ldots,u_l,v\}$ be the vertices on a shortest $u,v$ path. Then $P \cap R^{\circ}(S^1_f)= \emptyset$.
\end{lemma}
\begin{proof}
We will prove the above by contradiction. Let $P=\{u,u_1,u_2,\ldots,u_l,v\}$ be the minimum length counter example to the statement of the lemma, i.e.~$P$ is a shortest $u,v$ path with minimum $|P|$ such that $P \cap R^{\circ}(S^1_f) \neq \emptyset$. Minimality of $|P|$ implies that $P \cap S_f^1=\{u,v\}$. Let $T_{f}$ be the cycle created by replacing one of the two $u,v$ paths on the cycle $S^1_f$ by $P$ such that $R(f) \subseteq R(T_{f})$. Since $ P \cap R^\circ(S_f^1) \neq \emptyset, R(T_{f}) \subset R(S^1_{f})$. 

Suppose there exists $x,y \in f$ such that $R(f,x,y) \subsetneq R(T_f)$. Let $P_1[x,y]$ be the shortest $x,y$ path contained in $R(f,x,y)$. Since $R(f,x,y) \subseteq R(S^1_f)$, all the vertices of $P_1[x,y]$ are contained inside $R(S^1_f)$. Order the vertices on $P_1[x,y]$ from $x$ to $y$ and choose $x_1,y_1 \in P$ such that (i) all vertices in $P_1[x,x_1] \cup P_1[y_1,y]$ lie in $R(T_f)$ (ii) none of the vertices in $P_1[x_1,y_1]$ lie in $R^{\circ}(T_f)$. We observe that $P_1[x,x_1] \cup P[x_1,y_1] \cup P_1[y_1,y]$ is another shortest $x,y$ path contained inside $R(f,x,y)$, which contradicts the uniqueness of $P_1[x,y]$ (by Claim~\ref{clm:minimality}). Hence $R(f,x,y) \subseteq R(T_f)$ for all $x,y \in f$ and by Claim~\ref{clm:sf_rfuv}, $R(S_f^1) \subseteq R(T_f) \subset R(S_f^1)$, which is a contradiction.
\end{proof}

\begin{figure}[ht!]  
\centering
\includegraphics[width=5 in]{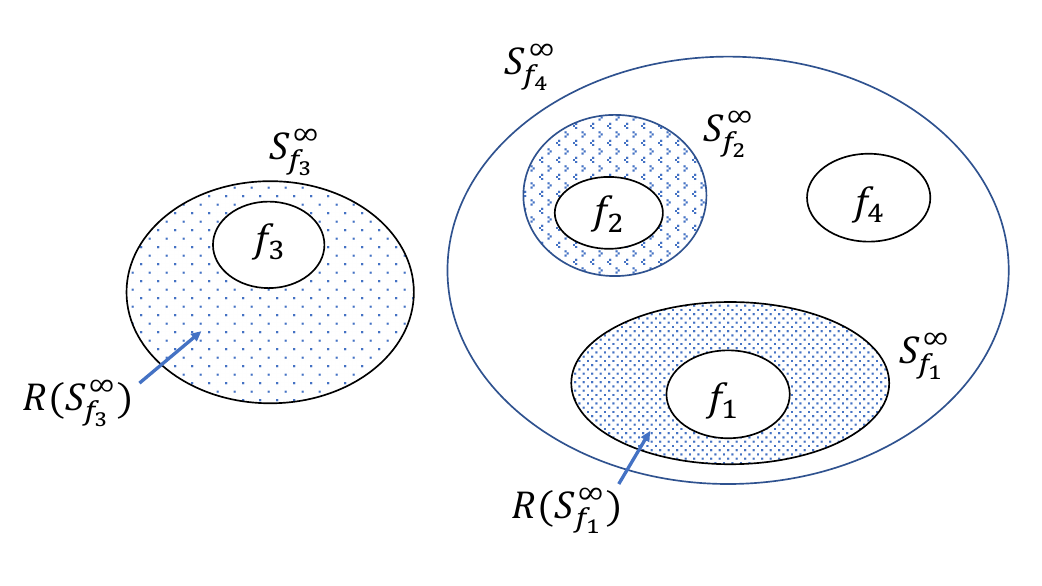}
\caption{Illustration for the statement of Lemma \ref{Lemma: face_support_laminar_final}.}
\label{fig:lemma_5}
\end{figure}

\begin{lemma}\label{Lemma: face_support_laminar_final}
The following are true for any distinct $f_1,f_2 \in F$:
\begin{enumerate}
    \item Either $R^\circ(S^\infty_{f_1}) \cap R^\circ(S^\infty_{f_2})= \emptyset$ or $R(S^\infty_{f_1}) \subseteq R(S^\infty_{f_2})$ or $R(S^\infty_{f_2}) \subseteq R(S^\infty_{f_1})$. 

    \item If $R(f_1) \subseteq R(S^\infty_{f_2})$, then $R(S^\infty_{f_1}) \subseteq R(S^\infty_{f_2}) \setminus R^{\circ}(f_2)$. Furthermore, $R(S^\infty_{f_1}) \neq R(S^\infty_{f_2})$.
\end{enumerate}
\end{lemma}
\begin{proof}
We will prove the above by using induction on the number of vertices. Suppose that the statement of the lemma holds for all graphs on less than $n$ vertices. If the graph doesn't contain any non-geodesic faces, then the statement of the lemma is trivially true. Lemma \ref{Lemma: face_support_laminar} states that the regions $\{R(S_f^1)\}_{f \in F}$ form a laminar family. If all the faces are not geodesic, then there exists $f \in F_N$ such that for any other $f' \in F_N$, either $R^{\circ}(S^1_f) \cap R^{\circ}(S^1_{f'}) = \emptyset$ or $R(S^1_f) \subset R(S^1_{f'})$. Furthermore, $R^{\circ}(f') \cap R^{\circ}(S^1_f)=\emptyset$. By induction, the statement of the lemma holds for the smaller graph $G \setminus R^{\circ}(S^1_f)$. Lemma \ref{Lemma:no_internal_path} states that no vertex in $R^{\circ}(S_f^1)$ lies on a shortest $u,v$ path for any $u,v \in S_f^1$. This implies that removing all the vertices in $R^{\circ}(S^1_{f})$ doesn't change the distance between any pair of vertices in $G\setminus R^{\circ}(S_f^1)$. By definition of $R(S_f^j)$, it immediately follows that for all $i>0$, $R(S_f^{i+1})$ in $G$ is precisely $R(S_{S^1_f}^i)$ in $G \setminus R^{\circ}(S^1_f)$. Hence, $R(S_f^{\infty})$ in $G$ is precisely $R(S_{S^1_f}^\infty)$ in $G \setminus R^{\circ}(S^1_f)$. By Lemma~\ref{Lemma:no_internal_path}, for all $f' \neq f$ such that $R^{\circ}(f') \cap R^{\circ}(S_f^1)=\emptyset$, $R(S_{f'}^i)$ in $G$ and $G \setminus R^{\circ}(S_f^1)$ are the same. By induction the lemma holds for the smaller graph $G\setminus R^{\circ}(S_1^f)$ and by the above observations, it follows that the lemma also holds for $G$.
\end{proof}
Let $\alpha > 1$ be a given constant. Given a simple cycle $C$, we say that $C$ is \textbf{$\alpha$-loose} if for any $u,v \in C$, the length of shortest $u,v$ path in $G[R^{\circ}(C) \cup \{u,v\}]$, which does not use any edge of $C$, is at least $\alpha \cdot d_{G}(u,v)$, i.e.~$d_{G[R^{\circ}(C) \cup \{u,v\}] \setminus E(C)} (u,v) \geq \alpha \cdot d_{G}(u,v)$. To make the induction argument work, only the laminar property of face supports is not sufficient and we need a stronger structure. We say that a length function $l:E \rightarrow \mathbb{R}_{\geq 0}$ is \textbf{$\alpha$-good} if there exists $\{S_f^\infty\}_{f\in F}$  such that $S^{\infty}_f$ is $\alpha$-loose for each $f \in F$. We next show that any length function can be converted into a $\alpha$-good one by modifying the edge-lengths by a multiplicative factor of at most $\alpha$.

\begin{theorem} \label{Theorem:alpha_good_decomposition}
Let $G=(V,E)$ be a planar graph with $f_I$ as the infinite face, edge-lengths $l :E \rightarrow \mathbb{R}_{\geq 0}$ and $\alpha >1$ be a constant. Then there exists new edge-lengths $l':E \rightarrow \mathbb{R}_{\geq 0}$ such that:
\begin{enumerate}
    \item $G$ is $\alpha$-good with respect to $l'$
    \item $l'(u,v)=l(u,v)$ for $u,v \in f_I$ 
    \item $\dfrac{l(e)}{\alpha}\leq l'(e) \leq l(e)$ for $e \in E$.
\end{enumerate}
\end{theorem}

\begin{proof}
We will prove the theorem by induction on the number of faces in the graph. If $|F|=1$ or $f_I$ is the only non-geodesic face, then $F_N= \emptyset$ and the statement of the theorem follows trivially by setting $l'=l$. Suppose that $G$ has at least one non-geodesic face other than $f_I$ w.r.t edge-length $l$. By Lemma \ref{Lemma: face_support_laminar_final}, regions in the set $\{R(S_f^\infty)\}_{f\in F}$ form a laminar family. A face $f \in F_N$ is called maximal if for any other face $f'$, either $R(S^\infty_{f'}) \subset R(S^\infty_f)$ or $R^{\circ}(S^\infty_{f'}) \cap R^{\circ}(S^\infty_f) = \emptyset$. Let $f_1,f_2,\dots,f_k \in F_N$ be the maximal faces of $G$ w.r.t $l$. Let $G_i = G[R(S_{f_i}^\infty)]$ and $F_i$ be the faces of $G_i$. If $k=1$ and $R(S_{f_1}) \neq R(f_I)$ or $k \geq 2$, then the number of faces in each $G_i$ is strictly less than $G$. By induction, for each $G_i$ (with infinite face $S^{\infty}_{f_i}$) we have edge-lengths $l_i$ satisfying the conditions of the theorem. We construct the new edge-length $l'$ as follows: if an edge $e$ is in $G_i$, we set $l'(e)=l_i(e)$; if an edge $e$ is not contained in any of the $G_i$, we set $l'(e)=l(e)$. Note that if an edge $(u,v)$ is contained in both $G_i$ and $G_j$, then it must be present on the infinite face of $G_i$ and $G_j$ and $l_i(u,v)=l_j(u,v)=l(u,v)$ (by property (2)). Hence $l'$ is well-defined. 

For the sake of better exposition, we will use $d_{H}^{x}(u,v)$ to denote the shortest path distance between vertices $u$ and $v$ in graph $H$ w.r.t edge-length $x$. From the definition of $G_i$, it follows that w.r.t edge-lengths $l$, $R(S^{\infty}_f) \subseteq R(S^{\infty}_{f_i})$ for each $f \in F_i$. Claim \ref{claim:S_f_geodesic} states that for each $u,v \in R(S_{f_i}^\infty)$, there is a $u,v$ shortest path using only the vertices in $R(S_{f_i}^\infty)$. This implies that $d_{G_i}^{l}(u,v)=d_{G}^{l}(u,v)$ for $u,v \in G_i$. By condition (2) of the induction hypothesis, $d_{G_j}^{l'}(u,v)=d_{G_j}^{l}(u,v)$ for all $u,v \in S_{f_j}^\infty$ and $j=1,2,\ldots,k$. Hence we conclude that $d_{G_i}^{l'}(u,v)=d_{G}^{l'}(u,v)$ for all $u,v \in G_i$. This implies that w.r.t edge-length $l'$, $R(S^{\infty}_f) \subseteq R(S^{\infty}_{f_i})$ for each $f \in F_i$. Hence for any $f \in F_i$, $R(S^{\infty}_f)$ (w.r.t edge-lengths $l'$) remains unchanged after gluing together different length functions for $G_1,G_2,\ldots,G_k$. Since $l'$ is $\alpha$-good for each $G_i$ (by the induction hypothesis), $l'$ is $\alpha$-good for $G$ as well. 

Suppose that $k=1$ and $R(S^{\infty}_{f_1}) = R(f_I)$. We consider two cases depending on whether the face $f_I$ is $\alpha$-loose or not. Suppose that the face $f_I$ is not $\alpha$-loose. By the definition of an $\alpha$-loose face, there must exist $u,v \in f_I$ such that the shortest $u,v$ path using no other vertices on $f_I$ has length $\beta \cdot d_G(u,v)$, for some $\beta < \alpha$. Let $P[u,v]=\{u,u_1,\ldots,u_l,v\}$ be such a path with minimum value of $\beta$. Note that $\beta=1$ is also possible. Let $l_1$ be the length function obtained by setting $l_1(e)=l(e)/\beta$ for all edges on the path $P[u,v]$ and $l_1(e)=l(e)$ otherwise.
\begin{claim} \label{claim:shortest_path}
For any $x,y \in f_I$, $l_1(x,y)=l(x,y)$.
\end{claim}
\begin{proof}
Suppose that the statement of the claim is false. By construction, $l_1(x,y) \leq l(x,y)$ for any $x,y \in f_I$. Let $x,y \in f_I$ be such that $l(x,y)>l_1(x,y)$ and $l_1(x,y)$ is minimum. Let $P_1[x,y]$ be the minimum length $x,y$ path w.r.t $l_1$. By minimality of $l_1(x,y)$, it follows that $P_1[x,y] \cap f_I=\{x,y\}$, i.e.~$P_1(x,y)$ uses only the vertices in $\{ x,y \} \cup R^{\circ}(f_I)$. Let the length of path $P_1(x,y)$ w.r.t $l$ be $l(P_1(x,y))$. Since $l_1$ is constructed by reducing the length of a subset of edges in $P_1(x,y)$ by a factor of $\beta$, we have that $l(P_1(x,y)) \leq \beta \cdot l_1(x,y) < \beta \cdot l(x,y)$. This contradicts the minimality of $P(u,v)$ and hence such a $x,y \in f_I$ cannot exist. 
\end{proof}
The path $P[u,v]$ divides $R(f_I)$ into two closed regions, say $R_1$ and $R_2$. Let $G_1=G[R_1]$ and $G_2=G[R_2]$. By Claim \ref{claim:shortest_path}, $P$ is a shortest $u,v$ path w.r.t $l_1$. This implies that for $u,v \in G_i$, $d_{G_i}(u,v)=d_{G}(u,v)$, for $i=1,2$. Hence for any $f \in G_i$, $R(S^{\infty}_f) \subseteq R_i$, for $i=1,2$. Therefore we can use induction to compute modified length functions for $G_1, G_2$ separately and combine them as before. More precisely, we inductively find length functions $l_1,l_2$ satisfying the conditions of the theorem and set $l'(e)=l_i(e)$ depending on whether edge $e$ lies in $G_1$ or $G_2$. If $e$ is in $G_1$ and $G_2$, then $e$ is on the infinite face of $G_1$ and $G_2$ and by the statement of the theorem, $l_1(e)=l_2(e)$, and $l'$ is well defined. By using the induction hypothesis, Claim \ref{claim:shortest_path} and arguments similar to the one used in the previous case, it immediately follows that the condition (i) and (ii) of the theorem is also satisfied. Hence, $G$ is $\alpha$-good w.r.t $l'$.

Suppose that the face $f_I$ is $\alpha$-loose and $S_{f_1}^\infty=f_I$. Then the statement of the theorem holds for non-geodesic face $f_1$. If there are no other non-geodesic faces in $R^\circ(f_I)$, then the theorem follows by setting $l'=l$. Otherwise, let $g_1,g_2,\ldots,g_k$ be the maximal (non-geodesic) faces contained inside $R^\circ(f_I)$. By Lemma \ref{Lemma: face_support_laminar_final}, $R^{\circ}(S^\infty_{g_i}) \cap R^{\circ}(S^\infty_{g_j})=\emptyset$ for $i \neq j$. Let $G_i=G[R(S_{g_i})]$. By induction hypothesis, we have length functions $l_1,\ldots,l_k$ that satisfy the theorem. We set $l'(e)=l(e)$ for any edge not contained inside any of the $G_i's$ and set $l'(e)=l_i(e)$ if $e$ is in $G_i$. The proof that $l'$ satisfies the conditions of the theorem follows the same arguments as in the previous cases and we omit it here.
\end{proof}
\section{Embedding For The Geodesic Pairs} \label{Section: embed_geodesic_pairs}
Let $G=(V,E)$ be a planar graph and $u,v \in V$ be such that $u,v \in f$ for some  $f \in F$. Furthermore, suppose that there exists a shortest $u,v$ path using only the edges of $f$, i.e.~$R(f,u,v)=R(f)$. We call $(u,v)$ a \textbf{geodesic pair}. Let $T$ be the set of all the geodesic pairs in $G$. In Theorem \ref{Theorem: embed_geodesic_pairs}, we show that there exists an embedding of $V$ into $L_1$ which preserves the distances between all the geodesic pairs within a constant factor. Using the congestion-distortion theorem, the following result follows directly from Theorem 12 of Kumar~\cite{kumar2020multicommodity}.
\begin{theorem} \label{Theorem: Separable Instances}
Let $G=(V,E)$ be a planar graph with edge length $l:E \rightarrow \mathbb{R}_{\geq 0}$ and $R$ be a set of pairs of vertices $(u,v)$ such that both $u$ and $v$ lie on the same face. Let $R_f \subseteq R$ denote the set of pairs of vertices incident on face $f$. Suppose for each $f \in F$, there exists disjoint set of vertices $X_1,Y_1,X_2,Y_2,\ldots,X_k,Y_k \subseteq f$ such that vertices in $X_i,Y_i$ and $X_i \cup Y_i$ appear contiguously on $f$ and for each $(u,v) \in R_f, u \in X_j,v \in Y_j$ for some $j \in \{1,2,\dots,k\}$. Then there exists $h: V \rightarrow L_1$ such that $||h(u)-h(v)|| \leq d_{G}(u,v)$ for all $(u,v) \in E$ and $||h(u)-h(v)|| \geq d_{G}(u,v)/3$ for all $(u,v) \in R$. 
\end{theorem}

\begin{figure}[ht] \label{thm_kumar}
\centering
\includegraphics[width=4 in]{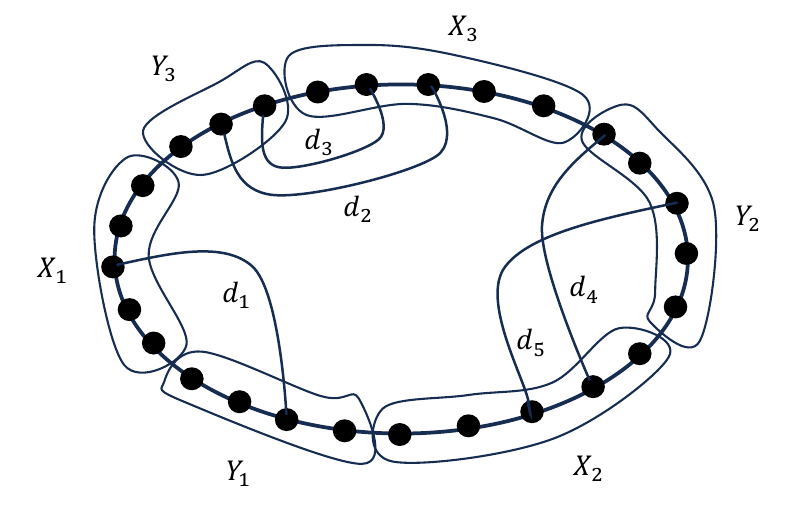}
\caption{The figure shows a face $f$ and the terms associated with it in Theorem~\ref{Theorem: Separable Instances}. $X_1,Y_1,X_2,Y_2,X_3,Y_3$ form a partition of the vertex set of face $f$ and $R_f$ consists of $d_1,d_2,d_3,d_4,d_5$.}
\end{figure}
\begin{theorem} \label{Theorem: embed_geodesic_pairs}
Let $G=(V,E)$ be a planar graph with edge length $l:E \rightarrow \mathbb{R}^{+}$. Then there exists $g:V \rightarrow L_1$ such that $||g(u)-g(v)|| \leq d_{G}(u,v)$ for all $(u,v) \in E$ and $||g(u)-g(v)|| \geq d_{G}(u,v)/21$ for all $(u,v) \in T$. 
\end{theorem}
\begin{proof}
We first construct a set of geodesic pairs $T_p \subseteq T$ as follows. Initially $T_p=\emptyset$. Let $f \in F$ and $V(f)=\{v_1,v_2,\ldots,v_{n},v_{n+1}=v_1\}$ be the vertices of $f$ in clockwise order. We partition $V(f)$ into contiguous sets of vertices as follows: let $k$ be the largest index such that $P_k=\{v_1,v_2,\ldots, v_k\}$ is a shortest $(v_1,v_k)$ path. We set $T_p=T_p \cup (v_1,v_k)$ and $C_{f}^{1}=\{v_1,v_2,\ldots,v_k\}$. We then start from $v_{k+1}$ and find the largest index $l$ such that $\{v_{k+1},\ldots,v_l\}$ is a shortest $(v_{k+1},v_l)$ path. We set $T_p=T_p \cup (v_{k+1},v_l)$ and $C_{f}^{2}=\{v_{k+1},v_{k+2},\ldots,v_l\}$. We continue with this process until all the vertices in $V(f)$ have been exhausted, i.e.~$V(f)=\cup_{i\geq 0}C_{f}^{i}$. We repeat the above for all the faces of $G$. By construction $G \cup T_p$ is planar. Therefore by Theorem \ref{seymour_embedding}, we have an embedding, say $h: V \rightarrow L_1$, that preserves the distance between pairs of vertices in $T_p$ exactly. In other words, $||h(u)-h(v)|| = d_{G}(u,v)$ for all $(u,v) \in T_p$ and $||h(u)-h(v)|| \leq d_{G}(u,v)$ for all $(u,v) \in E$. Let $T'=\{(u,v)|u,v \in C_{f}^{i}, f \in F, i \geq 0 \}$. By construction, first and last vertex (w.r.t the ordering of vertices on $V(f)$) of each of the set $C_{f}^{i}$ form a geodesic pair. Hence for all $(u,v) \in T'$, we have that $\||h(u)-h(v)||= d_{G}(u,v)$.
In Claim \ref{Claim: Partition Geodesic Pairs}, we show that the geodesic pairs in $T \setminus T'$ can be partitioned into $T_0, T_1,T_2,T_3$, $T_4,T_5$ such that each one of them satisfies the condition of Theorem \ref{Theorem: Separable Instances}. Then by Theorem \ref{Theorem: Separable Instances}, there exists $h_i:V\rightarrow L_1$ such that $||h_i(u)-h_i(v)|| \leq d_{G}(u,v)$ for $(u,v) \in E$ and $||h_i(u)-h_i(v)|| \geq d_{G}(u,v)/3$ for $(u,v) \in T_i$ for $i=0,1,2,3,4,5$. By setting $g=\dfrac{1}{7}(h+h_0+h_1+h_2+h_3+h_4+h_5)$, we obtain $||g(u)-g(v)|| \leq d_{G}(u,v)$ for $(u,v) \in E$ and $||g(u)-g(v)|| \geq d_{G}(u,v)/21$ for $(u,v) \in T$.
\begin{claim} \label{Claim: Partition Geodesic Pairs}
$T \setminus T'$ can be partitioned into $T_0,T_1,T_2,T_3,T_4,T_5$ such that each of the $T_i$'s satisfies the condition of Theorem \ref{Theorem: Separable Instances}.
\end{claim}
\begin{proof}
It is sufficient to show that the geodesic pairs incident on each face can be partitioned into $T^{f}_{i}$ for $i=0,1,2,3,4,5$ such that each $T^{f}_{i}$ satisfies the condition of Theorem \ref{Theorem: Separable Instances}. Let $f \in F$ and $C_{f}^{1},C_{f}^{2},\ldots, C_{f}^{k}$ be the partition of $V(f)$ created in the first step. Suppose there exists $(u,v) \in T \setminus T'$ such that $u \in C_{f}^{(l)}, v \in C_{f}^{(l+t)}$ for some $t \geq 3$, where $(p)=p$ if $p\leq k$ and $p-k$ otherwise. Then our procedure in the first step would not have created separate partitions for $C_{f}^{(l+1)}$ and $C_{f}^{(l+2)}$ and hence such a geodesic pair $(u,v)$ can't exist. Therefore for all $u,v \in C(f)$ and $(u,v) \in T \setminus T'$, one of the following must hold: $u \in C_{f}^{(l)}, v \in C_{f}^{(l+1)}$ (called \emph{type 1}) or $u \in C_{f}^{(l)}, v \in C_{f}^{(l+2)}$ (called \emph{type 2}) for some $l \in \{1,2,\ldots,k\}$.

Let $(u,v)$ be a geodesic pair $(u,v)$ such that $u \in C_{f}^{i}$ and $v \in C_{f}^{j}$. We say that $(u,v)$ belongs to $\emph{class}$ $i$ if $\min(i,j)~\mod 3=i$. Let $T^{f}_0,T^{f}_1,T^{f}_2$ be the geodesic pairs of class 0,1,2 of \emph{type 1} incident on the face $f$ and $T^{f}_3,T^{f}_4,T^{f}_5$ be the geodesic pairs of class 0,1,2 of \emph{type 2} incident on the face $f$. We set $T_i=\bigcup_{f \in F} T^{f}_i$ for $i=0,1,2,3,4,5$. The geodesic pairs in each of the $T_{i}$'s satisfy the condition of Theorem \ref{Theorem: Separable Instances} and this completes the proof of the claim.
\end{proof}  \end{proof}

\section{Single Source Shortest Path Embeddings} \label{Section:single_source}
In this section, we show that there exists a polynomial time computable embedding such that the distance between all pairs of vertices whose shortest path uses a fixed vertex is approximately preserved. To prove this result, we make use of a well known result of Klein, Plotkin and Rao \cite{klein1993excluded} on small diameter decomposition. Let $(X,D)$ be a finite metric space. A distribution $\mu$ over (vertex) partitions of $X$ is called $(\beta,\Delta)$-lipschitz if every partition $P$ in the support of $\mu$ satisfies $ S \in P \implies diam_{X}(S) \leq \Delta$ and moreover for all $x,y \in X$, $\underset{P \sim \mu}{\mathcal{P}}[P(x) \neq P(y)] \leq \beta \cdot \dfrac{d(x,y)}{\Delta}$. Klein, Plotkin and Rao \cite{klein1993excluded} showed that there exists a $(c,\Delta)$-lipschitz partition of a planar metric where $\Delta>0$ is arbitrary and $c$ is an absolute constant (independent of $\Delta$). In fact Klein, Plotkin and Rao \cite{klein1993excluded} showed that such a partition exists for any minor-closed family of graphs.

\begin{theorem}[Klein, Plotkin and Rao \cite{klein1993excluded}] \label{Theorem:KPR}
For any given finite planar metric $(X,D)$ and $\Delta >0 $, there exists a polynomial-time computable $(c,\Delta)$-lipschitz partition of $(X,D)$ for some absolute constant $c$ independent of $\Delta$.
\end{theorem}
\begin{theorem} \label{Theorem: single_source_embedding}
Let $G=(V,E)$ be a planar graph with edge-length $l:E \rightarrow \mathbf{R}_{\geq 0}$ and $v \in V$. Let $T=\{(s_i,t_i)\}_{i=1}^{k}$ be the set of all pair of vertices such that $d(s_i,t_i)=d(v,s_i)+d(v,t_i)$ for $i \in [k]$. Then there exists a polynomial time computable $g:V \rightarrow L_1$ and an absolute constant \footnote{we emphasize that $\beta$ is independent of $G$, in fact one can take $\beta=96 \cdot c$} $\beta \geq 1$ such that $d(s_i,t_i)/\beta \leq||g(s_i)-g(t_i)|| \leq d(s_i,t_i)$ for all $(s_i,t_i) \in T$.
\end{theorem}

\begin{proof}
We first prove the theorem for the special case when $d(v,s_i)=d(v,t_i)$ for all $(s_i,t_i) \in T$. Let $B_i=\{x|d(v,x) \leq 2^{i+1}\}$ for $i\geq 0$. Since $d(v,s_j)=d(v,t_j)$ for each $(s_j,t_j)\in T$, there exists an $i$ such that $(s_j,t_j) \in B_{i+1} \setminus B_i$. Let $T_i=\{(s_j,t_j)\in T |2^{i} \leq d(v,s_j)=d(v,t_j) \leq 2^{i+1} \}$ and $G_i$ be the graph obtained by setting the length of all the edges contained inside $Y_i=B_{i-2}$ and $Q_i=V \setminus B_{i+2}$ to 0. More formally, for all $(u,w)\in E$ such that $u,w \in Y_i$ or $u,w \in Q_i$, we set $l_{(u,w)}=0$ to obtain $G_i$ from $G$. Claim \ref{Claim:distrotion due to contraction} shows that distance between the $(s_i,t_i)$ pairs in $G_i$ are within a constant factor of the distance in $G$.
\begin{claim} \label{Claim:distrotion due to contraction}
$d_{G}(s_j,t_j)/4 \leq d_{G_i}(s_j,t_j) \leq d_{G}(s_j,t_j)$ for all $(s_j,t_j) \in T_i$.
\end{claim}
\begin{proof}
$d_{G_i}(s_j,t_j) \leq d_{G}(s_j,t_j)$ follows directly from construction since each $G_i$ is formed by setting the length of some edges of $G$ to 0. Since $(s_j,t_j) \in T_i$, we have $2^{i+1} \leq d(s_j,t_j)=d(s_j,v)+d(t_j,v) \leq 2^{i+2}$. If the shortest path between $(s_j,t_j)$ in $G_i$ doesn't use any vertex in $Y_i$ or $Q_i$, then $d(s_j,t_j)$ remains unchanged and the statement of claim follows trivially. Suppose the $(s_j,t_j)$ shortest path in $G_i$ uses a vertex in $Y_i$ or $Q_i$. Then we have:  
\begin{align*}
  d(Y_i,s_j)=d(B_{i-2},s_j) \geq d(B_{i-2},V \setminus B_{i-1})\geq 2^{i-1}\\
  d(Q_i,s_j)=d(V \setminus B_{i+2},s_j) \geq d(V \setminus B_{i+2},B_{i+1})\geq 2^{i+1}\\ d(Y_i,t_j)=d(B_{i-2},t_j) \geq d(B_{i-2},V \setminus B_{i-1})\geq 2^{i-1}\\
  d(Q_i,t_j)=d(V \setminus B_{i+2},t_j) \geq d(V \setminus B_{i+2},B_{i+1})\geq 2^{i+1}  
\end{align*}
We have $d_{G_i}(s_j,t_j) \geq \min\{ d(Y_i,s_j)+d(Y_i,t_j), d(Q_i,s_j)+d(Q_i,t_j) \}   \geq 2^{i} \geq \dfrac{d_{G}(s_j,t_j)}{4}$ and the claim follows. 
\end{proof}
We use Theorem \ref{Theorem:KPR} to construct a distribution over (vertex) partitions $\mathcal{P}$ of $G_i$ by setting $\Delta=2^{i-1}$. Suppose the partition is $(c,\Delta)$-lipschitz for some absolute constant $c>0$. We construct a cut metric $\mathcal{C}_i$ using $\mathcal{P}$ as follows: for each partition $P=\{P_1,P_2,\ldots,P_k\} \in \mathcal{P}$ with weight $\mu(P)$, we include $P_1,P_2,\ldots,P_k$ in $\mathcal{C}_i$, each with weight $\dfrac{\mu(P) \cdot \Delta}{c}$. Claim \ref{Claim: distortion for G_i} shows that $\mathcal{C}_i$ preserves distances for pairs in $T_i$.
\begin{claim} \label{Claim: distortion for G_i}
$\mathcal{C}_{i}(u,w) \leq d_{G_i}(u,w)$ for $u,w \in G_i$ and $\mathcal{C}_{i}(s_j,t_j) \geq \dfrac{d_{G_i}(s_j,t_j)}{4 \cdot c}$ for $(s_j,t_j) \in T_i$.
\end{claim}
\begin{proof}
By definition of $(c,\Delta)$-lipschitz partition, probability that $u$ and $w$ are in separate partitions is at most $\dfrac{c \cdot d_{G_i}(u,w)}{\Delta}$ for $u,w \in G_i$. Hence, for $u,w \in G_i$ we have:
\begin{equation*}
   \mathcal{C}_{i}(u,w)= \underset{P \sim \mu}{\mathcal{P}}[P(u) \neq P(w)] \cdot \dfrac{\Delta}{c} \leq \dfrac{c \cdot d_{G_i}(u,w)}{\Delta} \cdot \dfrac{\Delta}{c} \leq d_{G_i}(u,w).
\end{equation*}
For every partition $P \in \mathcal{P}$, $(s_j,t_j)$ are in different subsets of $P$ since $d_{G_i}(s_j,t_j) \geq 2^{i}> \Delta$. This implies that the total weight of cuts in $\mathcal{C}_i$ separating $(s_j,t_j)$ is at least $\dfrac{\Delta}{c}$. Hence,
\begin{equation*}
\mathcal{C}_{i}(s_j,t_j) \geq \dfrac{\Delta}{c} \geq \dfrac{2^{i-1}}{c} \geq \dfrac{d_{G_i}(s_j,t_j)}{4 \cdot c}    
\end{equation*}
The last inequality follows from Claim~\ref{Claim:distrotion due to contraction}. This completes the proof of the claim.
\end{proof}

\begin{claim} \label{Claim:final_distortion_distance_equal}
$\mathcal{C}(u,w) \leq 3 \cdot d_{G}(u,w)$ for $u,w \in G$ and $\mathcal{C}(s_j,t_j) \geq \dfrac{d_{G}(s_j,t_j)}{16 \cdot c}$ for $(s_j,t_j) \in T$, where $\mathcal{C}=\bigcup_{i \geq 0} \mathcal{C}_i$.
\end{claim}
\begin{proof}
By Theorem \ref{Theorem:KPR}, if an edge has length $0$ in $G_i$, then no $C \in \mathcal{C}_i$ contains exactly one end point of the edge. By construction, any edge $(u,w) \in E$ can be a part of at most three $G_i$'s. Hence, using Claim \ref{Claim: distortion for G_i}, we obtain $\mathcal{C}(u,w) \leq 3 \cdot d_{G}(u,w)$. For the second part, observe that:
\begin{equation*}
    \mathcal{C}_{i}(s_j,t_j) \geq \dfrac{d_{G_i}(s_j,t_j)}{4 \cdot c} \geq \dfrac{d_{G}(s_j,t_j)}{4 \cdot 4 \cdot c} = \dfrac{d_{G}(s_j,t_j)}{16 \cdot c}
\end{equation*}
The second and third inequality follow from Claim~\ref{Claim: distortion for G_i} and \ref{Claim:distrotion due to contraction} respectively, and this completes the proof of the claim.
\end{proof}
We have shown that the theorem holds when $d_{G}(v,s_j)=d_{G}(v,t_j)$ for all $(s_j,t_j) \in T$ (more precisely, the statement of the theorem holds for the cut-metric $\mathcal{C}/3$). We prove a more general version of Claim \ref{Claim:final_distortion_distance_equal} which will be useful in proving the theorem. 
\begin{claim} \label{Claim:general_claim_unequal}
$\mathcal{C}(s_j,t_j) \geq \dfrac{d_{G}(s_j,t_j)}{16 \cdot c}-3 \cdot |d_{G}(v,s_j)-d_{G}(v,t_j)|$ for $(s_j,t_j) \in T$.
\end{claim}
\begin{proof}
We may assume that $d_{G}(v,s_j) \geq d_{G}(v,t_j)$ w.l.o.g. Let $s'_j$ be the vertex on a shortest path from $v$ to $s_j$ such that $d_{G}(v,s'_j)=d_{G}(v,t_j)$. By Claim \ref{Claim:final_distortion_distance_equal}, we have $\mathcal{C}(s'_j,t_j) \geq d_{G}(s'_j,t_j)/16c$. All $C \in \mathcal{C}$ which contain exactly one of $s'_j$ and $t_j$ contribute to $\mathcal{C}(s'_j,t_j)$ and can be partitioned into two groups: one in which both $s'_j$ and $s_j$ are in $C$ and the other in which both $s_j$ and $t_j$ are in $C$. Let the sum of weights of the cuts in the two sets formed by this partition be denoted by $\mathcal{C}(s'_js_j,t_j)$ and $\mathcal{C}(s'_j,s_jt_j)$ respectively. Observe that $\mathcal{C}(s'_j,t_j)=\mathcal{C}(s'_js_j,t_j)+ \mathcal{C}(s'_j,s_jt_j)$ and $\mathcal{C}(s'_js_j,t_j) \leq \mathcal{C}(s_j,t_j)$. Hence $\mathcal{C}(s_j,t_j) \geq \mathcal{C}(s'_j,t_j)-\mathcal{C}(s'_j,s_jt_j)$. By Claim \ref{Claim:final_distortion_distance_equal}, we have,
\begin{equation*}
\mathcal{C}(s'_j,s_jt_j) \leq \mathcal{C}(s'_j,s_j) \leq 3d_{G}(s'_j,s_j)=3|d_{G}(v,s_j)-d_{G}(v,t_j)|    
\end{equation*}
Using $\mathcal{C}(s'_j,t_j) \geq \dfrac{d_{G}(s'_j,t_j)}{16 \cdot c}$, we obtain: 
\begin{equation*}
\mathcal{C}(s_j,t_j) \geq \mathcal{C}(s'_j,t_j)-3|d_{G}(v,s_j)-d_{G}(v,t_j)| \geq \dfrac{d_{G}(s_j,t_j)}{16 \cdot c}-3|d_{G}(v,s_j)-d_{G}(v,t_j)|. 
\end{equation*} 
This completes the proof of the claim. 
\end{proof}
We augment $\mathcal{C}$ by the following single source cuts: let $V=\{v_1,v_2,\ldots,v_n\}$ such that $d_G(v,v_1) \leq d_G(v,v_2)\leq \ldots \leq d_G(v,v_n)$. Let $R_i=\{v_1,v_2,\ldots,v_i\}$ and $w(R_i)=3 \cdot (d_{G}(v_{i+1})-d_{G}(v_{i}))$ for $i=1,2,3,\ldots,n$. Let $\mathcal{R}=\{(R_1,w(R_1)),\ldots,(R_n,w(R_n))\}$ and  $\mathcal{C}'=\mathcal{C} \cup \mathcal{R}$. 
\begin{claim}
$\mathcal{C}'(u,w) \leq 6 \cdot d_{G}(u,w)$ for $u,w \in G$ and $\mathcal{C}'(s_j,t_j) \geq \dfrac{d_{G}(s_j,t_j)}{16 \cdot c}$ for $(s_j,t_j) \in T$.
\end{claim}
\begin{proof}
Observe that $\mathcal{R}(u,w)=3 \cdot (d_{G}(v,u)-d_{G}(v,w)) \leq 3 \cdot d_{G}(u,w)$. Using Claim \ref{Claim:general_claim_unequal}, for any $u,w \in G$ we have, $\mathcal{C}'(u,w)=\mathcal{C}(u,w) + \mathcal{R}(u,w) \leq 3 \cdot d_{G}(u,w)+ 3 \cdot d_{G}(u,w)=6 \cdot d_{G}(u,w)$. Using Claim \ref{Claim:general_claim_unequal}, for any $(s_j,t_j) \in T$ we have: 
\begin{equation*}
 \mathcal{C}'(s_j,t_j)=\mathcal{C}(s_j,t_j) + \mathcal{R}(s_j,t_j)\geq \dfrac{d_{G}(s_j,t_j)}{16 \cdot c}-3|d_{G}(v,s_j)-d_{G}(v,t_j)|+3|d_{G}(v,s_j)-d_{G}(v,t_j)|=\dfrac{d_{G}(s_j,t_j)}{16 \cdot c}   
\end{equation*}
This completes the proof of the claim.
\end{proof}
Let $\mathcal{C''}=\mathcal{C'}/6$. Then $\mathcal{C''}(u,w) \leq d_{G}(u,w)$ for $u,w \in G$ and $\mathcal{C''}(s_j,t_j) \geq \dfrac{d_{G}(s_j,t_j)}{96 \cdot c}$ for $(s_j,t_j) \in T$. By using the equivalence between the cut-metric and $L_1$-embedding, we have the desired $g : V \rightarrow L_1$ with $\beta=96 \cdot c$.
\end{proof}
\section{Constrained Embedding} \label{Section:constrained_embedding}
Let $G=(V,E)$ be a planar graph and $f$ be its infinite face. Let $V(f)=\{v_1,v_2,\ldots,v_{r}=v_1\}$ be the vertices on the cycle defined by $f$ in clockwise order. Suppose that we are given a cut-metric $\mathcal{C}=\{(C_1,w_1),(C_2,w_2),\ldots,(C_m,w_m)\}$ w.r.t $V(f)$ such that each cut $C_i \subseteq V(f)$ corresponds to a contiguous subset of vertices on the cycle defined by $f$. We say that $\mathcal{C}$ is a \textbf{cut-metric w.r.t $f$}. Suppose that we wish to extend $\mathcal{C}$ to $V$, i.e.~we wish to find a cut-metric $\mathcal{D}$ w.r.t $V$ such that $\mathcal{D}= \cup_{i=1}^{m}\mathcal{D}_i$, $\mathcal{D}_i=\{(D^1_i, w^1_i), (D^2_i,w^2_i),\ldots, (D^k_i,w^{k_i}_i)\}$, $\sum_{j=1}^{k_i}w^j_i=w_i$, $C_i \subseteq D^j_i$, $(V(f) \setminus C_i)  \cap D^j_i =\emptyset$ for $i=1,2,\ldots,m$, and $\mathcal{D}(u,v) \leq d_{G}(u,v)$ for $(u,v) \in E$. We call $\mathcal{D}$ an \textbf{extension} of $\mathcal{C}$ to $G$ w.r.t face $f$. Lemma \ref{Lemma:extend_face_to_rest} shows that this is always possible if $\mathcal{C}(u,v) \leq d_{G}(u,v)$ for $u,v \in V(f)$. Note that by definition of $\mathcal{D}$, it follows that $\mathcal{D}(u,v)=\mathcal{C}(u,v)$ for all $u,v \in V(f)$.  
\begin{lemma} \label{Lemma:extend_face_to_rest}
Let $G=(V,E)$ be a planar graph and $f$ be its infinite face. Let $\mathcal{C}$ be a cut-metric w.r.t face $f$ such that $\mathcal{C}(u,v) \leq d_{G}(u,v)$ for all $u,v \in f$ and $C_i \in \mathcal{C}$ corresponds to a contiguous set of vertices on $f$. Then there exists an extension $\mathcal{D}$ of $\mathcal{C}$ to $G$ w.r.t $f$.
\end{lemma}
\begin{proof}
We set up a multicommodity flow instance in the planar dual of $G$ such that all the sink-source pairs are on the infinite face and the cut-condition is satisfied. We use Theorem \ref{OS} to find a feasible flow and the fact that (simple) cycles in a planar graph correspond to (central) cuts in the planar dual to finish the proof.
Let $\mathcal{C}=\{(C_1,w_1),(C_2,w_2),\ldots,(C_m,w_m)\}$, $V(f)=\{v_1,v_2,\ldots,v_r=v_1\}$ be the vertices on cycle of $f$ and $e_i=(v_i,v_{i+1}), 1 \leq i \leq r-1$ be the edges of $f$. Let $G^{D}=(V^{D},E^{D})$ be the planar dual of $G$ and $f^D$ be the dual vertex corresponding to the infinite face $f$. For each $e^D=(u^D,v^D) \in E^D$, we set the capacity of edge $e^D$ as $d_G(u,v)$, i.e.~$c(e^D):=d_G(u,v)$. Let $e^{D}_{1},\ldots, e^{D}_{r-1}$ be the edges incident on $f^D$ in $G^D$. We split the vertex $f^D$ into $r-1$ vertices $f^D_{1},\ldots, f^{D}_{r-1}$ such that $e^{D}_i$ is the only edge incident on vertex $f^{D}_i$. We overload the notation and use $G^D$ to denote this modified graph as well, and jump between the two implicitly when the context is clear. Note that each of the vertices $f^{D}_i$ lie on a single face of $G^D$.

By theorem statement, each $C_i \in \mathcal{C}$ is a contiguous subset of vertices on $f$. Hence, each $C_i \in \mathcal{C}$ is incident on exactly two edges on $f$, say $e_j,e_k$. We set up a multicommodity flow instance in $G^D$ as follows: for each $C_i \in \mathcal{C}$, we introduce a demand edge $(f^D_j,f^D_k)$ with demand value $w_i$. By Lemma \ref{central}, to check that the cut-condition is satisfied for the instance, we only need to verify it for central cuts. Recall that a cut $S$ is central if both $G[S]$ and $G[V \setminus S]$ are connected. Recall that each central cut in $G^D$ corresponds to a simple cycle in $G$. Also, we need to check the cut-condition only for central-cuts in $G^D$ which separate some non-zero demand. All such central cuts in $G^D$ are incident to the face containing all the demands, and hence we can conclude that each central cut in $G^D$ corresponds to a $(v_i,v_j)$ path in $G$ for some $1 \leq i <j \leq r-1$. Total capacity of the supply edges across such a cut is at least the length of shortest path between $(v_i,v_j)$ in $G$, i.e.~$d_{G}(v_i,v_j)$. The total demand across such a cut is equal to $\mathcal{C}(v_i,v_j)$. Since $\mathcal{C}(v_i,v_j) \leq d_{G}(v_i,v_j)$ for each $v_i,v_j \in V(f)$, the cut-condition is satisfied and we have a feasible flow satisfying all the demands. Each flow path in $G^D$ corresponds to a set of edges in $G$, which in turn correspond to a cut in $G$ (recall that a cycle in a planar graph corresponds to a cut in its (planar) dual). Let $\mathcal{D}$ be the set of cuts defined by a feasible solution to the multicommodity flow instance. Since the total flow through any edge $e^D=(u^D,v^D)$ in $G^D$ is at most $c(u^D,v^D)=d_G(u,v)$, the total weight of cuts in $\mathcal{D}$ separating $e=(u,v)$ is at most $d_G(u,v)$. Hence $\mathcal{D}$ is a valid extension of $\mathcal{C}$ and this completes the proof of the lemma.   
\end{proof}
\begin{theorem} \label{Theorem:extend_the_embedding}
Let $G=(V,E)$ be a planar graph with edge-length $l:E \rightarrow \mathbf{R}_{\geq 0}$, $F$ be its set of faces and $\alpha \geq 12 \beta$ be an absolute constant \footnote{$\beta$ is the constant from Theorem \ref{Theorem: single_source_embedding}}. Furthermore, let $S$ be an $\alpha$-loose cycle such that $f_2 \in F$ is the unique non-geodesic face contained inside $R(S)$. Let $F_1$ be the set of faces of $G[R(S)]$. Let $\mathcal{C}$ be a cut-metric w.r.t $S$ such that $d_{G}(u,v) / \alpha \leq \mathcal{C}(u,v) \leq d_{G}(u,v)$ for $u,v \in S$. Then there exists an extension of $\mathcal{C}$ to $G[R(S)]$, say $\mathcal{Z}$, such that $d_{G}(u,v)/ \alpha \leq \mathcal{Z}(u,v) \leq d_{G}(u,v)$ for all $u,v \in f \in F_1$.
\end{theorem}

\begin{proof}
We abuse notation and use $S$ to denote the set of vertices and edges incident on the cycle $S$. Consider two copies of $G[R(S)]=(V_1,E_1)$, say $H_1$ and $H_2$ with length functions $l_1$ and $l_2$ defined as follows: $l_1(u,v):= l(u,v)$ for all $(u,v) \in S$ and $l(u,v)/\alpha$ otherwise; $l_2(u,v):= 0$ for all $(u,v) \in S$ and $l(u,v)$ otherwise. 

By the definition of an $\alpha$-loose cycle, any $u,v$ path (with at least 3 vertices) $P[u,v] \subseteq R(S)$ such that $P[u,v] \cap S=\{u,v\}$ has length at least $\alpha \cdot d_G(u,v)$. The shortest $u,v$ path might go inside and outside of $R(S)$ multiple times. We break this path into sub-paths whose interior vertices are completely contained inside or outside of $R(S)$ respectively. We then use the fact that $S$ is an $\alpha$-loose cycle and the length of edges with at least one end-point inside $R^{\circ}(S)$ have been scaled down by a multiplicative factor of $\alpha$ in $H_1$ to conclude that $d_{H_1}(u,v) \geq d_G(u,v)$ for any $u,v \in S$. Lemma \ref{Lemma:extend_face_to_rest} shows that the cuts in $\mathcal{C}$ can be extended to $H_1$ such that $\mathcal{C}(u,v) \leq d_{H_1}(u,v)$ for each $(u,v) \in E_1$. Let this cut-metric be $\mathcal{C}'$. We create an (equivalent) $L_1$ embedding $h:V_1 \rightarrow L_1$ from $\mathcal{C}'$ by forming a new coordinate for each $C \in \mathcal{C}'$ and setting $h(u)=0$ if $u \in C$ and $w_C$ otherwise.

We now consider $H_2$ with edge length $l_2$. Since the length of all the edges in $S$ have been set to zero in $l_2$, we may treat all the vertices on the cycle $S$ as a single node, say $v_S$. If $d_{H_2}(u,v) \neq d_G(u,v)$ for some $u,v \in V_1$, then it must be the case that $d_{H_2}(u,v) = d_{H_2}(u,v_S)+d_{H_2}(v_S,u)$. Let $T_1$ be the set of all pairs of vertices $(u,v)$ in $H_2$ such that the shortest path between $u$ and $v$ in $H_2$ uses the vertex $v_S$. We use Theorem \ref{Theorem: single_source_embedding} to find an embedding $g_1:V_1 \rightarrow L_1$ such that $d_{H_2}(u,v)/\beta \leq ||g_1(u)-g_1(v)|| \leq d_{H_2}(u,v)$ for all $(u,v) \in T_1$.

Let $T_2$ be the set of all geodesic pairs in $H_2$ (see Section \ref{Section: embed_geodesic_pairs} for the definition of geodesic pairs). In this case, we use Theorem \ref{Theorem: embed_geodesic_pairs} to find an embedding $g_2:V_1 \rightarrow L_1$ such that $d_{H_2}(u,v)/21 \leq ||g_2(u)-g_2(v)|| \leq d_{H_2}(u,v)$ for all $(u,v) \in T_2$. 

Let $T_3$ be the set of all pairs of vertices $(u,v)$ such that $u,v \in f_2$ i.e.~the set of all pairs of vertices on the non-geodesic face $f_2$. We use Theorem \ref{OS_embedding} to find an embedding $g_3:V_1 \rightarrow L_1$ such that  $||g_3(u)-g_3(v)||=d_{H_2}(u,v)$ for all $(u,v) \in T_3$ and  $||g_3(u)-g_3(v)|| \leq d_{H_2}(u,v)$ for all $(u,v) \in E_1$. Let $g=(g_1+g_2+g_3)/3$. Since $||g_i(u)-g_i(v)|| \leq d_{H_2}(u,v)$ for $i=1,2,3$, we have $||g(u)-g(v)|| \leq d_{H_2}(u,v)$ for $(u,v) \in E_1$. 

Let $T$ be the set of all pair of vertices $u,v$ which lie on the same face, i.e.~$T=\{ (u,v)|u,v \in f  \text{ for some} f \in F_1\}$. Since $f_2$ is the unique non-geodesic face in $F_1 \setminus S$, for any $u,v \in T$, one of the following must hold: (i) any shortest path between $u,v$ in $H_2$ uses the vertex $v_S$ (ii) $(u,v)$ is a geodesic pair (iii) $u,v \in f_2$. Hence, $T \subseteq T_1 \cup T_2 \cup T_3$ and we have  $||g(u)-g(v)|| \geq d_{H_2}(u,v)/\beta_1$ where $\beta_1 = \max \{3 \cdot \beta,3 \cdot 21, 3 \cdot 1\}=3 \beta$. Let $z:= h + \dfrac{\alpha -1 }{\alpha} \cdot g$. 
\begin{claim}
$||z(u)-z(v) || \leq l(u,v)$ for $(u,v) \in E_1$. 
\end{claim}
\begin{proof}
Consider an edge $(u,v) \in S$. Since $l_1(u,v)=l(u,v)$ and $l_2(u,v)=0$ we have:
\begin{equation*}
 ||z(u)-z(v)||=||h(u)-h(v)|| +\dfrac{\alpha-1}{\alpha} \cdot ||g(u)-g(v)|| \leq l_1(u,v)+\dfrac{\alpha-1}{\alpha}  \cdot l_2(u,v)=l(u,v)   
\end{equation*}
Now consider an edge $(u,v) \in E_1 \setminus S$. Since $l_1(u,v)=\dfrac{l(u,v)}{\alpha}$ and $l_2(u,v)= l(u,v)$,
\begin{equation*}
||z(u)-z(v)||=||h(u)-h(v)|| +\dfrac{\alpha-1}{\alpha}\cdot ||g(u)-g(v)||\leq \dfrac{1}{\alpha} \cdot l(u,v)+\dfrac{\alpha-1}{\alpha } \cdot l(u,v)=l(u,v).
\end{equation*}
\end{proof}

\begin{claim}
$||z(u)-z(v) || \geq \dfrac{d_G(u,v)}{\alpha}$ for $(u,v) \in T$.
\end{claim}
\begin{proof}
Let $(u,v) \in T$. We consider two cases depending on whether the shortest $u,v$ path in $H_2$ uses the vertex $v_S$. If the shortest $u,v$ path doesn't use the vertex $v_S$, then we have $d_{G}(u,v)=d_{H_2}(u,v)$ and,
\begin{equation*}
 ||z(u)-z(v)|| \geq \dfrac{ \alpha-1}{\alpha} \cdot ||g(u)-g(v)|| \geq \dfrac{1}{2} \cdot \dfrac{d_{H_2}(u,v)}{3 \beta} \geq \dfrac{d_G(u,v)}{\alpha}   
\end{equation*}
Suppose that the shortest $u,v$ path uses the vertex $v_S$. Recall that $v_S$ was formed by identifying all the vertices in $S$ as a single vertex. We uncontract $v_S$ and let $u-u_1-v_1-v$ be the shortest $u,v$ path where $u_1,v_1 \in S$ are the first and last vertices of $S$ on the shortest path. Note that the shortest $u_1$-$v_1$ path may contain some of the vertices in $G \setminus G_1$. Since $\mathcal{C}(u_1,v_1) \geq d_G(u_1,v_1)/\alpha$, we have, 
\begin{equation*}
    ||h(u)-h(v)|| \geq \dfrac{d_{G}(u_1,v_1)}{\alpha}-d_{H_1}(u_1,u)-d_{H_1}(v_1,v).
\end{equation*} 
Using the fact that $\alpha \geq 12 \beta=4 \beta_1$, we have: \\

$||z(u)-z(v)|| \geq \dfrac{d_{G}(u_1,v_1)}{\alpha}-d_{H_1}(u_1,u)-d_{H_1}(v_1,v)+\dfrac{\alpha-1}{\alpha} \cdot \dfrac{d_{H_2}(u,u_1)+d_{H_2}(v_1,v)}{\beta_1}$\\

$ \hspace{27 mm} \geq \dfrac{d_{G}(u_1,v_1)-d_{G}(u_1,u)-d_{G}(v_1,v)}{\alpha}+ \dfrac{1}{2} \cdot \dfrac{d_{G}(u,u_1)+d_{G}(v_1,v)}{\beta_1}$\\

$\hspace{27 mm} \geq \dfrac{d_{G}(u,u_1)+d_{G}(u_1,v_1)+d_{G}(v_1,v)}{\alpha}$\\

$\hspace{27 mm} = \dfrac{d_{G}(u,v)}{\alpha}$. \\

To go from the fist equation to the second, we have used $d_{H_1}(u_1,u)= d_{G}(u_1,u)/\alpha$ and $d_{H_1}(v_1,v)= d_{G}(v_1,v)/\alpha$. This follows from the fact that we scale down the lengths of edges with atleast one end-point inside $R^{\circ}(S)$ by a multiplicative factor of $\alpha$ in $H_1$ and the shortest paths between $(u_1,u)$ and $(v_1,v)$ use only such edges.
\end{proof}
Using the equivalence of cut-metric and $L_1$-embedding, we can construct a cut-metric $\mathcal{Z}$ from $z: V_1 \rightarrow L_1$ satisfying the conditions of the theorem. 
\end{proof}
\section{Putting Everything Together}\label{Section:final}
\begin{theorem} \label{thm:almost_final}
 Let $G=(V,E)$ be a planar graph with an $\alpha$-good edge-length $l:E \rightarrow \mathbb{R}_{\geq 0}$. Let $F$ be its set of faces and $T=\{(u,v)|u,v \in f \in F\}$. Then there exists a polynomial time computable $z: V \rightarrow L_1$ such that $||z(u)-z(v)|| \leq l(u,v)$ for $(u,v) \in E$ and $||z(u)-z(v)|| \geq d_{G}(u,v)/\alpha$ for $(u,v) \in T$.
\end{theorem}
\begin{proof}
We prove the theorem by using induction on the number of vertices. If there are no non-geodesic faces, then we use Theorem \ref{Theorem: embed_geodesic_pairs} to get an $L_1$ embedding with distortion at most 21. If there exists a non-geodesic face, we first embed $G$ in the plane such that the infinite face is geodesic (such a face always exists). By Lemma~\ref{Lemma: face_support_laminar_final}, there exists a non-geodesic face $f$ such that for any non-geodesic face $f' \neq f$, either $R(S_f^{\infty}) \subset R(S_{f'}^{\infty})$ or $R^{\circ}(S_f^{\infty}) \cap R^{\circ}(S_{f'}^{\infty})=\emptyset$. Let $G_1=(V_1,E_1)= G \setminus R^{\circ}(S_f^\infty)$ and $G_2 = (V_2,E_2) = G[R(S_f^\infty)]$. By the definition of $\alpha$-good length function, $l$ is $\alpha$-good for $G_1$ as well.

Since $f$ is non-geodesic, there exists a vertex in $R^{\circ}(S^\infty_f)$. Hence, the number of vertices in $G_1$ is strictly smaller than $G$ and we can use induction. By induction, there exists an embedding $z_1:V_1 \rightarrow L_1$ satisfying the conditions of the theorem, i.e.~$||z_1(u)-z_1(v)|| \leq l(u,v)$ for $(u,v) \in E_1$ and $||z_1(u)-z_1(v)|| \geq d_{G_1}(u,v)/\alpha$ for $(u,v) \in T_1$, where $T_1$ is the set of pair of vertices on the same face in $G_1$. Using the equivalence between $L_1$-embedding and cut-metric, we compute a cut-metric equivalent to $z_1$, say $\mathcal{Z}_1$. 

Let $\mathcal{Z}_1^{f}$ be the cut-metric induced by $\mathcal{Z}_1$ on the cycle $S^{\infty}_f$. We use Theorem \ref{Theorem:extend_the_embedding} to compute a cut-metric which extends $\mathcal{Z}_1^{f}$ to vertices in $G_2$, say $\mathcal{Z}_f$. We obtain the final cut-metric by setting $\mathcal{Z}=(\mathcal{Z}_1 \setminus \mathcal{Z}_1^{f}) \cup \mathcal{Z}_f$. Let $z: V \rightarrow L_1$ be the equivalent embedding to $\mathcal{Z}$. By Theorem \ref{Theorem:extend_the_embedding}, $||z(u)-z(v)|| \leq l(u,v)$ for $(u,v) \in E$ and $||z(u)-z(v)|| \geq d_G(u,v)/ \alpha$ for $(u,v) \in T$.
\end{proof}

\begin{theorem}
Let $G=(V,E)$ be a planar graph with edge-lengths $l:E \rightarrow \mathbb{R}_{\geq 0}$. Let $F$ be the set of its faces and $T=\{(u,v)|u,v \in f \in F\}$. Then there exists a polynomial time computable $z: V \rightarrow L_1$ such that $||z(u)-z(v)|| \leq l(u,v)$ for $(u,v) \in E$ and $||z(u)-z(v)|| \geq d_{G}(u,v)/ \alpha^{2}$ for $(u,v) \in T$.
\end{theorem}
\begin{proof}
We first compute an $\alpha$-good length function $l'$ by using Theorem \ref{Theorem:alpha_good_decomposition}. We then use Theorem \ref{thm:almost_final} to compute $z: V \rightarrow L_1$ w.r.t $l'$. The statement of the theorem follows by noting that $l'$ is created from $l$ by reducing the length of each edge by a multiplicative factor of at most $\alpha$.
\end{proof}

\section{Conclusions} \label{Section:conclusions}
In this paper, we proved a $\mathcal{O}(1)$ flow-cut gap when $G$ is planar and both end points of every demand edge are incident on one of the faces. Although our result does not directly imply any bounds on the (half)-integral flow-cut gap, we believe that it should be possible to exploit the laminar structure of flows in such instances to prove such a bound. Inductive arguments have been used successfully for proving better flow-cut gaps for planar instances, for example series-parallel graphs \cite{chakrabarti2008embeddings} and $k$-outer planar graphs \cite{chekuri2006embedding}. We believe that the techniques developed in this paper could be useful for extending such an approach to a more general setting. Our current bound on the flow-cut gap for instances considered in this paper are quite large ($ \sim 10^8$), but the best lower bound that we know of is 4/3. It is an interesting open question to close this gap. 


\bibstyle{main}
\bibliography{main.bib}

\end{document}